\documentclass[12pt,english]{article}
\usepackage{textcomp}
\usepackage[utf8]{inputenc}
\usepackage{xcolor}
\usepackage{colortbl}
\usepackage{babel}
\usepackage{float}
\usepackage{amsmath}
\usepackage{amsthm}
\usepackage{amssymb}
\usepackage{geometry}
\usepackage{setspace}
\usepackage[bookmarks=false,
 breaklinks=false,pdfborder={0 0 1},backref=section,colorlinks=true]
 {hyperref}
\hypersetup{
 citecolor=black,linkcolor=black,urlcolor=black}

\makeatletter

\newtheorem{theorem}{Theorem}[section]

\makeatother

\begin{document}
\title{Kinetic metric for basins of attraction of RNA secondary structures and analysis of the ultrametricity of the energy landscape}
\author{ A.\,P.~Zubarev \\
 \textit{Volga State Transport University,} \\
 \textit{Pervyi Bezymyannyi pereulok 18, Samara, 443066, Russia;}
\\
 \textit{Samara University,} \\
 \textit{Moskovskoe shosse 34, Samara, 443123 Russia} \\
 e-mail:\thickspace{}\texttt{apzubarev@mail.ru} }
\maketitle
\begin{abstract}
A method for testing the hypothesis of an ultrametric organization of the energy landscape of RNA secondary structures is proposed, based on the analysis of the kinetics of transitions between basins of attraction. The method is based on a kinetic metric constructed from the spectral decomposition of the symmetrized Kramers transition rate matrix and the Mahalanobis distance, and it is not an ultrametric by construction. A computational scheme has been developed that includes automatic filtering of noise eigenmodes and a procedure for analyzing disconnected structure graphs. The performance of the method is demonstrated on a sample of reference and random RNAs.
\end{abstract}

\section{Introduction}

\label{sec_intr}

The hierarchical organization of complex systems, in which elements are combined into groups, groups into larger groups, and so on, finds a natural mathematical description in terms of ultrametric spaces. A review of the mathematical foundations of ultrametricity and its applications is given in \cite{rammal1986}. Recall that a metric $d$ on a set $M$ is called an ultrametric if for any three points $x,y,z\in M$ the strong triangle inequality holds: $d(x,y)\le\max\{d(x,z),d(z,y)\}$. In such a space, any two open balls either do not intersect, or one is contained within the other, and every finite ultrametric space is isomorphic to the set of leaves of a rooted tree with a suitable height function \cite{Hartigan}. This property makes the ultrametric a central tool for hierarchical clustering \cite{Sneath} and phylogenetic analysis \cite{Lake}.

A turning point in the physics of disordered systems was the discovery of the ultrametric structure of the state space in the Sherrington-Kirkpatrick spin glass model \cite{parisi1979,parisi1980,mezard1984,mezard1987}. Rigorous mathematical proofs of this fact were obtained later \cite{talagrand2003,talagrand2011,panchenko2013}. It turned out that ultrametricity is not an artificial construction but reflects the real geometry of the energy landscape of systems with many competing minima. Almost simultaneously with the development of spin glass theory, Frauenfelder and co-workers, studying the kinetics of carbon monoxide binding to myoglobin, put forward a hypothesis about the hierarchical organization of the conformational space of proteins \cite{frauenfelder1988,frauenfelder1991,frauenfelder2000}. According to this hypothesis, a protein can exist in a multitude of metastable substates, which are grouped into clusters separated by energy barriers of increasing height. The dynamics of the protein is thereby interpreted as a random walk on an ultrametric tree of substates. Further development of the ideas of ultrametricity led to the application of $p$-adic analysis methods to the description of the energy landscapes of biopolymers \cite{ALL,ALL_1,DKKM}. In \cite{avetisov2002,ABZ_2014,avetisov2009,bikulov2021}, a connection was established between the hierarchy of conformational states of proteins and $p$-adic parametrization, and ultrametric diffusion on such landscapes was investigated. In parallel, $p$-adic methods found application in other areas of mathematical biology: in \cite{khrennikov2007,kozyrev2008,dragovich2009,dragovich2011,dragovich2021}, a $p$-adic model of the genetic code was proposed, and it was shown that many properties of nucleotide sequences, including mutational stability and evolutionary trajectories, allow an interpretation in terms of $p$-adic and ultrametric geometry.

The secondary structure of ribonucleic acids (RNA) represents another class of systems for which the question of the hierarchical organization of the energy landscape is of fundamental importance. As is well known, an RNA molecule consists of a chain of four types of nucleotides (adenine, cytosine, guanine, uracil), which can pair with each other via hydrogen bonds according to the Watson-Crick complementarity principle \cite{brion1997,draper1999,tinoco1971,crothers1974}. The formation of such pairs lowers the free energy of the molecule \cite{freier1986,mathews1999}. The secondary structure of RNA is a set of nucleotide pairs in which each chain position participates in at most one pair, and the pairs themselves do not form pseudoknots \cite{zuker1989,hofacker1994}. The functional repertoire of RNA is exceptionally broad: in addition to classical messenger, transfer, and ribosomal RNA, there are numerous non-coding RNAs performing catalytic, regulatory, and structural functions \cite{golden1998,kiss2002,hainzl2002}. For each of them, the formation of a specific secondary structure is a necessary condition for biological activity, and the folding pathway of the molecule into this structure is determined by the geometry of its energy landscape.

The first study of the ultrametric properties of the energy landscape of RNA secondary structures was undertaken by Higgs \cite{higgs1996}. In a simplified model where only canonical pairs with a unit energy contribution are allowed and no loop penalties exist, it was found that the distribution of overlaps between low-energy states reveals features of an ultrametric organization. This approach was further developed in the joint work of Morgan and Higgs \cite{Morgan1998}, where an explicit procedure for calculating energy barriers between ground states was proposed. A fundamental limitation of this procedure, however, is that the final distance between states is determined by the single-linkage method, which by construction generates an ultrametric. Thus, the question of whether ultrametricity is a real property of the RNA landscape or an artifact of the measurement method remains open.

In our previous work \cite{zubarev2026}, within the framework of a combinatorial model of RNA secondary structures (the Nussinov algorithm \cite{nussinov1978}), the ultrametric properties of the set of ground state configurations (global energy minima) were numerically investigated. It was illustrated that for a number of selected reference small nuclear RNAs (snRNA), the degree of ultrametricity of the set of ground states, determined based on the Jaccard metric \cite{levandowsky1971}, varies widely and can reach 80\% for specially optimized sequences. The results of that work show that within the framework of the model used, the degree of ultrametricity is a characteristic that depends very strongly on the order of nucleotides.

The present work is largely methodological. Its main contribution is the construction of a physically grounded kinetic metric on the space of basins of attraction of RNA secondary structures and the development of a complete computational scheme for its calculation. Here, we do not set the systematic investigation of the ultrametricity of specific RNA sequences as the primary task. The proposed metric is constructed on the set of thermodynamic states of RNA based on the analysis of the kinetics of transitions between structures using the Kramers formalism and the spectral decomposition of the symmetrized transition rate matrix. Such a distance takes into account all possible transition paths and temperature factors, making the verification of ultrametricity a meaningful task whose result is determined by the real structure of the energy landscape. We also present the results of our computational experiments on a sample consisting of a set of representatives of small nuclear RNAs. These results are illustrative in nature and serve two purposes: to demonstrate the performance of the method on real sequences with experimentally determined thermodynamic parameters and to reveal characteristic qualitative regularities that should be considered when planning future systematic studies. We have shown that the degree of nontrivial ultrametricity of the proposed kinetic metric varies widely from approximately $43\%$ to $89\%$ for the selected set of RNAs. For most of the investigated sequences, the structure graph generated by stochastic sampling is essentially disconnected; in this regard, a fragmentation index is introduced in the work, and analysis by individual connectivity components is applied. The data obtained indicate that, although strict hierarchical organization is not observed for any of the investigated RNAs, the energy landscapes possess a partial ultrametric structure, the manifestation of which is specific to each sequence. Individual representatives demonstrate strong ultrametricity (more than 80\%), and a comparative analysis with null models shows that, for a fixed nucleotide composition, the degree of nontrivial ultrametricity is determined by the order of their sequence.

The present article is organized as follows. Section \ref{sec_math} introduces the basic definitions: primary and secondary structure of RNA, free energy, adjacency relation, structure graph, attractor points, basins of attraction, single-linkage distance on the structure graph. Section \ref{sec_spectral} constructs the Kramers transition rate matrix, describes the procedure for its symmetrization preserving the spectrum, and demonstrates its negative semi-definiteness. Section \ref{sec_mahalanobis} constructs the kinetic metric, namely, the Mahalanobis distance between basins of attraction is introduced, and it is proved to be a metric on the corresponding factor space, but not automatically an ultrametric. Section \ref{sec_filter} describes the algorithm for filtering noise modes. Section \ref{sec_ultr} outlines the methodology for determining the degree of ultrametricity of a metric. Section \ref{sec_computational} provides a complete computational scheme for analyzing the degree of ultrametricity of the kinetic metric on the set of basins of attraction of RNA secondary structures, corresponding to the software implementation \cite{zubarev2026_github}. This section also discusses methods for statistical testing of null hypotheses about ultrametricity. Section \ref{sec_results} presents illustrative computational results for 6 reference RNAs from the database \cite{refseq}. Section \ref{sec_conclusion} discusses the advantages, disadvantages, and limitations of the proposed approach, as well as directions for its further development.

\section{Mathematical model of the RNA secondary structure system}

\label{sec_math}

In this section, we will give a mathematical description of the system to be analyzed. Let us define a set of four symbols (nucleotides):
\[
\mathcal{A}=\{\mathrm{A},\mathrm{U},\mathrm{G},\mathrm{C}\}.
\]
The symbol $\mathrm{A}$ denotes adenine, $\mathrm{U}$ denotes uracil, $\mathrm{G}$ denotes guanine, $\mathrm{C}$ denotes cytosine. Let $n$ be a natural number. The primary structure of RNA of length $n$ is a mapping
\[
S:\{1,2,\dots,n\}\to\mathcal{A}.
\]
The elements of the set $\{1,2,\dots,n\}$ are called positions, and the number $n$ is the length of the primary structure. Two positions $i$ and $j$ ($i<j$) are called complementary if the ordered pair $(S(i),S(j))$ belongs to the set
\[
\mathcal{C}=\{(\mathrm{A},\mathrm{U}),(\mathrm{U},\mathrm{A}),(\mathrm{G},\mathrm{C}),(\mathrm{C},\mathrm{G}),(\mathrm{G},\mathrm{U}),(\mathrm{U},\mathrm{G})\}.
\]
The set of all ordered pairs of positions $(i,j)$ with $i<j$ is denoted by $\Pi$:
\[
\Pi=\{(i,j)\in\mathbb{N}^{2}:1\le i<j\le n\}.
\]

As is well known, in statistical physics, a microstate of a molecular system is the complete atomic configuration of the system (the positions of all atoms in space). We do not consider such configurations and do not model them explicitly. Instead, we work with macrostates -- sets of microstates that share the same values of certain macroscopic parameters. In this work, two levels of macrostates are introduced -- structures and basins of attraction. An elementary macrostate (or secondary structure) on the primary structure $S$ is a set $P\subseteq\Pi$ satisfying the following four conditions:

(1) Complementarity condition: for each pair $(i,j)\in P$, positions $i$ and $j$ are complementary, i.e., $(S(i),S(j))\in\mathcal{C}$.

(2) Non-crossing condition: there are no two distinct pairs $(i,j),(k,l)\in P$ such that $i<k<j<l$.

(3) Minimum distance condition: for each pair $(i,j)\in P$, the inequality $j-i\ge h_{\min}+1$ holds, where $h_{\min}\ge3$ is the minimum hairpin loop length (the number of unpaired nucleotides). This condition guarantees the presence of a sufficient number of unpaired nucleotides between paired positions for hairpin loop formation. The standard value is $h_{\min}=3$.

(4) No multiple bonds condition: each position $k\in\{1,\dots,n\}$ is included in at most one pair from $P$.

In the following, the term ``structure'' will be used precisely for such elementary macrostates. Note that each structure $p$ represents a macrostate in the sense that it combines a set of atomic microstates sharing the same secondary structure. The set of all structures on the primary structure $S$ is denoted by $\Omega(S)$.

Each structure $p\in\Omega(S)$ is assigned a real number -- its Gibbs free energy. This work uses the nearest neighbor model \cite{Turner}, which allows the free energy of any secondary structure to be calculated. In this model, each structure $p\in\Omega(S)$ is uniquely decomposed into non-overlapping elementary structural components, which are defined as follows. A structure $p$ can be viewed as a graph whose vertices are the positions $\{1,\dots,n\}$ and edges are the pairs from $p$. The following types of connected components are distinguished in this graph: (1) a stack, which is two ``nested'' pairs $(i,j)$ and $(i+1,j-1)$; (2) a hairpin loop, which is a pair $(i,j)$ where all positions from $i+1$ to $j-1$ are unpaired and whose length $L=j-i-1$ satisfies $L\ge h_{\min}$; (3) an internal loop, which consists of two pairs $(i,j)$ and $(k,l)$ with $i<k<l<j$, where all positions between $i$ and $k$, as well as between $l$ and $j$, are unpaired, and there are no other pairs inside; (4) a multiloop, which is a connected component containing more than two pairs that is not a stack, hairpin loop, or internal loop; (5) a dangling end, which is a set of consecutive unpaired positions at one of the ends of the sequence (from position $1$ to the first paired position or from the last paired position to $n$). Each component $t$ is assigned tabulated enthalpy $\Delta H(t)$ and entropy $\Delta S(t)$ parameters, which also depend on the nucleotide sequence within the component. The free energy of a component at temperature $T$ is calculated by the formula
\[
\Delta G(t,T)=\Delta H(t)-T\Delta S(t).
\]
The parameters $\Delta H$ and $\Delta S$ are obtained experimentally and are tabulated \cite{Turner}. In computer calculations, access to these parameters is implemented by the ViennaRNA library \cite{Vienna}. The free energy of a structure $p\in\Omega(S)$ is the quantity
\[
G(p,T)=\sum_{t\in\text{comp}(p)}\Delta G(t,T),
\]
where $\text{comp}(p)$ is the set of elementary structural components that make up the structure $p$. The values $G(p,T)$ are expressed in kcal/mol, i.e., they are molar free energies.

To analyze the free energy function, it is necessary to introduce the concept of proximity between different structures. This is done using an adjacency relation, which is defined through elementary operations that transform one structure into another. Let $p,q\in\Omega(S)$ be two structures. We say that $p$ and $q$ are adjacent if one of them can be obtained from the other by one of the following three operations: (1) deletion of a pair: there exists a pair $(i,j)\in p$ such that $q=p\setminus\{(i,j)\}$; (2) addition of a pair: there exists a pair $(i,j)\in\Pi\setminus p$ such that $q=p\cup\{(i,j)\}$ and the set $q$ satisfies the conditions for a structure; (3) substitution of a pair: there exist pairs $(i_{1},j_{1})\in p$ and $(i_{2},j_{2})\in\Pi\setminus p$ such that $q=(p\setminus\{(i_{1},j_{1})\})\cup\{(i_{2},j_{2})\}$ and the set $q$ satisfies the conditions for a structure. For the adjacency relation of structures $p$ and $q$, we use the notation $p\leftrightarrow q$. To describe the adjacency relation on the set of structures, it is convenient to use a structure graph, which is an undirected graph $G=(\Omega(S),\mathcal{E})$, where $\Omega(S)$ is the set of vertices (structures), and $\mathcal{E}$ is the set of edges defined by the adjacency relation:
\[
\mathcal{E}=\{\{p,q\}\subseteq\Omega(S):\:p\leftrightarrow q,\:p\neq q\}.
\]

Traditionally, a local minimum of the free energy function $G$ is a structure $p\in\Omega(S)$ for which, for any adjacent structure $q\in\Omega(S)$ (i.e., for any $q$ such that $\{p,q\}\in\mathcal{E}$), the strict inequality $G(p,T)<G(q,T)$ holds. The set of all such strict local minima is denoted by $M(S)\subseteq\Omega(S)$. However, in real energy landscapes, regions of constant free energy -- ``plateaus'' -- where adjacent structures have the same energy, are not uncommon. A structure on a plateau does not satisfy the definition of a strict local minimum, but it is impossible to transition from it to a structure with lower energy. To correctly describe such situations physically, we introduce the concept of an attractor point.

Let a small positive threshold $\varepsilon_{\text{eq}}>0$ be given, accounting for possible numerical errors in calculating the free energy. An attractor point is a set of structures $A\subseteq\Omega(S)$ such that: (1) for any structure $p\in A$, there is no adjacent structure $q\in\Omega(S)$ with strictly lower free energy, i.e., $G(q,T)\ge G(p,T)-\varepsilon_{\text{eq}}$ for all $q$ such that $\{p,q\}\in\mathcal{E}$; (2) the set $A$ is a connected component of the graph whose vertices are all structures satisfying condition (1), and edges connect those structures $p$ and $q$ for which $\{p,q\}\in\mathcal{E}$ and $|G(p,T)-G(q,T)|<\varepsilon_{\text{eq}}$. Thus, an attractor point can be either an isolated strict local minimum (the set $A$ consists of one structure) or a plateau (a connected set of structures of approximately equal energy, from which it is impossible to transition to a structure with significantly lower energy). The collection of all attractor points is denoted by $\mathcal{A}(S)$. Each attractor point $A\in\mathcal{A}(S)$ is a connected component of the graph whose vertices are all structures that do not have a neighbor with strictly lower energy (taking into account the threshold $\varepsilon_{\text{eq}}$), and edges connect adjacent structures with approximately the same free energy.

Next, we define a gradient descent process that assigns a certain attractor point to each structure. This process is implemented by a sequence $\{p_{k}\}^{\infty}_{k=0}$, which is constructed recursively. Let $p\in\Omega(S)$ be an arbitrary structure. Set $p_{0}=p$. Further, if $p_{k}$ belongs to some attractor point (i.e., $p_{k}\in A$ for some $A\in\mathcal{A}(S)$), the sequence terminates. If $p_{k}$ does not belong to any attractor point, then there exists at least one adjacent element $q$ with $G(q,T)<G(p_{k},T)-\varepsilon_{\text{eq}}$, and we choose from all such elements the one for which the value $G(q,T)$ is minimal. If there are several such elements, we choose the one with the smallest index in a fixed numbering. Since the set $\Omega(S)$ is finite and at each step the free energy strictly decreases (taking the threshold into account), the sequence cannot be infinite and always reaches some attractor point in a finite number of steps.

Further, we define the mapping $\Phi:\Omega(S)\to\mathcal{A}(S)$ as follows. For any structure $p\in\Omega(S)$, the value $\Phi(p)$ is equal to the attractor point reached as a result of the process described above of sequential transition to the neighbor with the lowest free energy.

Let $A\in\mathcal{A}(S)$ be an attractor point. The basin of attraction corresponding to $A$ is the set
\[
\mathcal{M}(A)=\{p\in\Omega(S):\Phi(p)=A\}.
\]
The family $\{\mathcal{M}(A):A\in\mathcal{A}(S)\}$ forms a partition of $\Omega(S)$. Basins of attraction represent larger macroscopic states of the system, each of which combines a set of structures that ``flow down'' to one attractor point (a strict minimum or a plateau).

For thermodynamic analysis, it is necessary to introduce weights characterizing the probability of realization of each basin of attraction at a given temperature. Let an absolute temperature $T>0$ be given. The Boltzmann weight of a structure $p\in\Omega(S)$ is
\[
w(p)=\exp\left(-\frac{G(p,T)}{RT}\right),
\]
where $R$ is the universal gas constant ($R=N_{\mathrm{A}}k_{\mathrm{B}}$, where $N_{\mathrm{A}}$ is Avogadro's number, $k_{\mathrm{B}}$ is Boltzmann's constant, and the free energy $G(p,T)$ is expressed in molar units (kcal/mol)). The partition function of a basin of attraction $\mathcal{M}\subseteq\Omega(S)$ is the quantity
\[
Z(\mathcal{M})=\sum_{p\in\mathcal{M}}w(p)=\sum_{p\in\mathcal{M}}\exp\left(-\frac{G(p,T)}{RT}\right).
\]
The free energy of the basin of attraction is
\[
F(\mathcal{M})=-RT\ln Z(\mathcal{M}).
\]
\noindent Let us describe the construction of the distance based on the single-linkage method, used in \cite{Morgan1998}, applied to basins of attraction in terms of the definitions introduced above. A path from structure $p$ to structure $q$ is a finite sequence $\gamma=(p_{0},p_{1},\dots,p_{L})$, where $p_{0}=p$, $p_{L}=q$ and $\{p_{k},p_{k+1}\}\in\mathcal{E}$ for all $k$. Let $h(\gamma)=\max_{0\le k\le L}G(p_{k},T)$. For two structures, the distance between them is defined by the rule
\[
\rho(p,q)=\min_{\gamma\in\Gamma(p,q)}h(\gamma),
\]
where $\Gamma(p,q)$ is the set of all paths from $p$ to $q$. This distance is transferred to the basins of attraction $\mathcal{M}_{a},\mathcal{M}_{b}$ in a standard way:
\begin{equation}
\rho(\mathcal{M}_{a},\mathcal{M}_{b})=\min_{p\in\mathcal{M}_{a},\,q\in\mathcal{M}_{b}}\rho(p,q)\label{rho}
\end{equation}
for $a\neq b$ and $\rho(\mathcal{M}_{a},\mathcal{M}_{a})=0$. It is well known and directly verifiable that the function (\ref{rho}) is an ultrametric on the set of basins of attraction. Indeed, symmetry and non-degeneracy are obvious. Let us verify the strong triangle inequality. For three basins $\mathcal{M}_{a},\mathcal{M}_{b},\mathcal{M}_{c}$, choose paths $\gamma_{ac}$ and $\gamma_{cb}$ realizing $\rho(\mathcal{M}_{a},\mathcal{M}_{c})$ and $\rho(\mathcal{M}_{c},\mathcal{M}_{b})$. Since the endpoint of $\gamma_{ac}$ and the starting point of $\gamma_{cb}$ lie in the same basin $\mathcal{M}_{c}$, they can be connected by a path $\gamma_{\text{inner}}$ lying entirely within $\mathcal{M}_{c}$, whose height does not exceed the maximum of the heights of the two already chosen paths. The concatenation $\gamma_{ac}\circ\gamma_{\text{inner}}\circ\gamma_{cb}$ yields a path from $\mathcal{M}_{a}$ to $\mathcal{M}_{b}$, whose height is $\max\{\rho(\mathcal{M}_{a},\mathcal{M}_{c}),\rho(\mathcal{M}_{c},\mathcal{M}_{b})\}$. Since $\rho$ is the minimum height over all paths, we obtain the strong triangle inequality $\rho(\mathcal{M}_{a},\mathcal{M}_{b})\le\max\{\rho(\mathcal{M}_{a},\mathcal{M}_{c}),\rho(\mathcal{M}_{c},\mathcal{M}_{b})\}$. We emphasize that the ultrametricity of (\ref{rho}) does not depend on the specific form of the function $G$ or on the structure of the adjacency graph, but is a direct consequence of the fact that the distance is defined as the minimum over paths of the maximum energy on the path. This construction is equivalent to computing the single-linkage distance \cite{Sneath} on a graph whose vertices are structures and the edge weight $\{p,q\}$ is set to $\max\{G(p,T),G(q,T)\}$. Therefore, we call the distance (\ref{rho}) the ``single-linkage distance on the structure graph''. Note that it is precisely the single-linkage distance that underlies the derivation of the final barrier matrix in \cite{Morgan1998}; therefore, its ultrametricity is not an independent physical result. For a meaningful test of the hypothesis of an ultrametric organization of the energy landscape, a metric is needed that is not reducible to the construction (\ref{rho}) and does not rely on the single-linkage procedure. Such a metric is constructed in Sections \ref{sec_spectral} and \ref{sec_mahalanobis}.

\section{Transition rate matrix and its spectral decomposition}

\label{sec_spectral}

The proposed approach is based on describing the evolution of the system as a Markov process on the structure graph, assuming that the trajectories of this process represent random paths between structures. Let $f_{p}(t)$ be the probability of finding the RNA in structure $p$ at time $t$. Then the dynamics of $f_{p}(t)$ can be modeled by the master equation
\[
\dfrac{df_{p}(t)}{dt}=\sum_{q,\:\{p,q\}\in\mathcal{E}}\left(k_{q\to p}f_{q}(t)-k_{p\to q}f_{p}(t)\right),
\]
where only transitions between adjacent structures are considered, and $k_{p\to q}$ is the rate constant for the transition from structure $p$ to adjacent structure $q$.

According to the theory of chemical reaction rates \cite{Kramers,Onsager}, the rate constant for a transition between states is determined by the height of the energy barrier separating these states. In the secondary structure model, the exact geometry of the energy surface is inaccessible; therefore, a lower bound on the barrier is used through the energies of the initial and final states. For two adjacent structures $p$ and $q$, the rate constant for the transition $p\to q$ has the form
\[
k_{p\to q}=\nu_{0}\exp\left(-\frac{\max\{G(p,T),G(q,T)\}-G(p,T)}{RT}\right),
\]
where $\nu_{0}>0$ is a frequency factor (pre-exponential factor) with dimension {[}time{]}$^{-1}$, and the quantity $\max\{G(p,T),G(q,T)\}$ is used as a lower bound for the free energy of the transition state. This estimate is standard for models of RNA secondary structure kinetics, where the exact geometry of the energy surface is inaccessible \cite{Zuckerman}, and its applicability for elementary operations on secondary structures (addition/deletion of a single pair) is substantiated in \cite{hofacker1994,flamm2000}. It is obvious that the constants $k_{p\to q}$ satisfy detailed balance
\[
k_{p\to q}\cdot w(p)=k_{q\to p}\cdot w(q)
\]
where $w(p)=\exp(-G(p,T)/RT)$ is the Boltzmann weight of structure $p$.

Next, we define a matrix $K$ of size $N\times N$, where $N=|\Omega(S)|$, as follows:
\begin{equation}
K_{pq}=\begin{cases}
\nu_{0}\exp\left(-\dfrac{\max\{G(p),G(q)\}-G(q)}{RT}\right), & \text{for }\{p,q\}\in\mathcal{E},\\[10pt]
0, & \text{for }\{p,q\}\notin\mathcal{E},\;p\neq q,
\end{cases}\label{Kr}
\end{equation}
\[
K_{pp}=-\sum_{q\neq p}K_{qp}.
\]
Note that in the convention adopted here, the element $K_{pq}$ ($p\neq q$) represents the rate of transition from state $q$ to state $p$, as a result of which the sum of the elements of each column of the matrix $K$ is zero ($\sum_{p}K_{pq}=0$), and the generator acts on the column vector of probabilities $\mathbf{f}$ from the left according to the equation $\dot{\mathbf{f}}=K\mathbf{f}$. The matrix $K$ is the generator of a continuous-time Markov process. It satisfies the probability conservation condition and detailed balance $K_{pq}\,w(q)=K_{qp}\,w(p)$, where $w(p)=\exp(-G(p,T)/RT)$ is the stationary distribution. The matrix $K$ is non-symmetric, which complicates the numerical analysis of its spectrum, defined by the standard eigenvalue equation $K\phi=\lambda\phi$. A standard technique is to switch to a symmetric matrix $S$ that has the same eigenvalues as the matrix $K$. We define $S$ via a similarity transformation:
\begin{equation}
S=W^{-1/2}KW^{1/2},\label{eq_similarity}
\end{equation}
where $W=\mathrm{diag}(w(1),\dots,w(N))$ is the diagonal matrix of Boltzmann weights. For $p\neq q$, the elements of matrix $S$ are
\begin{equation}
S_{pq}=\frac{1}{\sqrt{w(p)}}\,K_{pq}\,\sqrt{w(q)}=\nu_{0}\exp\left(-\frac{|G(p)-G(q)|}{2RT}\right)\;\text{for }\{p,q\}\in\mathcal{E},\label{eq_S_offdiag}
\end{equation}
and $S_{pq}=0$ for non-adjacent vertices. The diagonal elements are preserved under this transformation:
\begin{equation}
S_{pp}=K_{pp}=-\sum_{q\neq p}K_{qp}.\label{eq_S_diag}
\end{equation}
The matrix $S$ thus constructed is symmetric and has a spectrum exactly coinciding with the spectrum of $K$. However, $S$ is not a Markov process generator in the strict sense, since the sums of the elements in its columns are not zero. The stationary distribution of the original process (eigenvalue $\lambda_{0}=0$) corresponds to an eigenvector of $S$ of the form $\psi_{0}=W^{1/2}\mathbf{1}$, where $\mathbf{1}$ is a vector of ones.

The connection between the spectra of matrices $K$ and $S$ is established directly through the similarity transformation. Substituting the eigenvector substitution $\phi=W^{1/2}\psi$ into the original equation $K\phi=\lambda\phi$ and multiplying on the left by $W^{-1/2}$ leads to the standard symmetric eigenvalue problem:
\begin{equation}
S\psi=\lambda\psi.\label{eq_standard_eigenvalue}
\end{equation}
Consequently, all eigenvalues $\lambda_{k}$, and hence the relaxation times $\tau_{k}=-1/\lambda_{k}$, are preserved upon transition from the original process to the symmetrized description. The replacement of the matrix $K$ by the matrix $S$ is motivated solely by the computational efficiency of diagonalizing symmetric matrices. The eigenvalues and, therefore, the weights $1/|\lambda_{k}|$ in the metric are preserved. The eigenvectors $\psi_{k}$ of the symmetrized matrix $S$ are orthonormal in the standard Euclidean inner product.

It is easy to see that the matrix $S$ is negative semi-definite. Indeed, consider an arbitrary vector $x\in\mathbb{R}^{N}$ and expand the quadratic form for the symmetric matrix $S$ using the definitions of its elements (\ref{eq_S_offdiag})--(\ref{eq_S_diag}):
\[
x^{T}Sx=\sum_{p}S_{pp}x^{2}_{p}+\sum_{p\neq q}S_{pq}x_{p}x_{q}
\]
\begin{equation}
=-\sum_{p}\left(\sum_{q\neq p}K_{qp}\right)x^{2}_{p}+\sum_{\{p,q\}\in\mathcal{E}}2S_{pq}x_{p}x_{q}.\label{xSx}
\end{equation}
Using the relationship between the elements of $S$ and $K$, we write $S_{pq}=\frac{1}{\sqrt{w(p)}}K_{pq}\sqrt{w(q)}=\frac{1}{\sqrt{w(q)}}K_{qp}\sqrt{w(p)}$, whence $K_{qp}=S_{pq}\sqrt{w(q)/w(p)}$. Substituting this expression into (\ref{xSx}), we obtain after elementary transformations:
\[
x^{T}Sx=-\sum_{\{p,q\}\in\mathcal{E}}S_{pq}\left(\sqrt{\frac{w(q)}{w(p)}}x^{2}_{p}+\sqrt{\frac{w(p)}{w(q)}}x^{2}_{q}-2x_{p}x_{q}\right)
\]
\begin{equation}
=-\sum_{\{p,q\}\in\mathcal{E}}S_{pq}\left(\sqrt[4]{\frac{w(q)}{w(p)}}x_{p}-\sqrt[4]{\frac{w(p)}{w(q)}}x_{q}\right)^{2}\le0.\label{eq_negative_semidef}
\end{equation}
Inequality (\ref{eq_negative_semidef}) holds due to the non-negativity of all off-diagonal elements $S_{pq}$ ($p\neq q$). Equality to zero is achieved if and only if $x_{p}/\sqrt{w(p)}=x_{q}/\sqrt{w(q)}$ for all pairs $\{p,q\}\in\mathcal{E}$, i.e., when the vector $x$ is proportional to $\sqrt{w}$ on each connected component of the graph $G$.

From negative semi-definiteness, it follows that $S$ has a complete set of eigenvalues $\lambda_{0},\lambda_{1},\dots,\lambda_{N-1}\in\mathbb{R}$ satisfying the condition $0=\lambda_{0}\ge\lambda_{1}\ge\lambda_{2}\ge\dots\ge\lambda_{N-1}$. The quantities $\tau_{k}=-1/\lambda_{k}$ ($k\ge1$) have the dimension of time and are relaxation times. Here, the smallest non-zero eigenvalues in absolute value correspond to the slowest processes -- transitions between the deepest basins of attraction, separated by high energy barriers.

\section{Distance between basins of attraction}

\label{sec_mahalanobis}

To construct a metric on the set of basins of attraction, we use the Mahalanobis distance, which was introduced in \cite{Mahalanobis} and is defined as follows. Let $X$ and $Y$ be two vectors in $\mathbb{R}^{m}$, and let $\Sigma$ be a positive definite matrix of size $m\times m$. The Mahalanobis distance between $X$ and $Y$ is the quantity
\begin{equation}
D_{\Sigma}(X,Y)=\sqrt{(X-Y)^{T}\Sigma^{-1}(X-Y)}.\label{Mah}
\end{equation}
If $\Sigma$ is the identity matrix, then the distance (\ref{Mah}) coincides with the Euclidean distance. In general, the matrix $\Sigma$ is interpreted as the covariance matrix of the data distribution, and the Mahalanobis distance accounts for correlations between the components of the vectors.

In our problem, the distance between basins is constructed based on the projections of their characteristic vectors onto the eigenvectors of the symmetrized matrix $S$. It is fundamentally important here that the original physical process is described by the non-symmetric matrix $K$, and directly projecting characteristic vectors onto the eigenvectors $\psi_{k}$ of matrix $S$ would be incorrect. Below, we provide a full mathematical justification for exactly how the projections should be computed to obtain a metric corresponding to the original Markov process.

The original matrix $K$ satisfies the detailed balance condition: $K_{pq}w(q)=K_{qp}w(p)$, where $w(p)=\exp(-G(p,T)/RT)$. We introduce the diagonal matrix of stationary weights $W=\mathrm{diag}(w(1),\dots,w(N))$. The detailed balance condition is equivalent to the symmetry of the matrix $KW$: $(KW)^{T}=KW$. This means that the matrix $K$ is similar to the symmetric matrix $S=W^{-1/2}KW^{1/2}$ and, consequently, is diagonalizable with a real spectrum. For matrix $K$, consider the right eigenvectors $\boldsymbol{\phi}_{k}$ satisfying $K\boldsymbol{\phi}_{k}=\lambda_{k}\boldsymbol{\phi}_{k}$, and the left eigenvectors $\tilde{\boldsymbol{\phi}}_{k}$ satisfying $\tilde{\boldsymbol{\phi}}^{T}_{k}K=\lambda_{k}\tilde{\boldsymbol{\phi}}^{T}_{k}$. From the detailed balance condition, it follows that the left and right vectors are related by a simple algebraic relation:
\begin{equation}
\tilde{\boldsymbol{\phi}}_{k}=W^{-1}\boldsymbol{\phi}_{k}.\label{eq_left_right_relation}
\end{equation}
Indeed, transposing the equality $K\boldsymbol{\phi}_{k}=\lambda_{k}\boldsymbol{\phi}_{k}$, we obtain $\boldsymbol{\phi}^{T}_{k}K^{T}=\lambda_{k}\boldsymbol{\phi}^{T}_{k}$. Multiplying by $W^{-1}$ and using $K^{T}W^{-1}=W^{-1}K$ (which is equivalent to detailed balance), we arrive at $\boldsymbol{\phi}^{T}_{k}W^{-1}K=\lambda_{k}\boldsymbol{\phi}^{T}_{k}W^{-1}$, i.e., $\tilde{\boldsymbol{\phi}}^{T}_{k}K=\lambda_{k}\tilde{\boldsymbol{\phi}}^{T}_{k}$ with $\tilde{\boldsymbol{\phi}}_{k}=W^{-1}\boldsymbol{\phi}_{k}$. The right and left eigenvectors form a biorthogonal system: $\langle\tilde{\boldsymbol{\phi}}_{k},\boldsymbol{\phi}_{\ell}\rangle=\delta_{k\ell}$. An arbitrary initial probability distribution $\mathbf{f}(0)$ is expanded in terms of the right eigenvectors with coefficients computed via the inner product with the left eigenvectors. The physical projection of the macrostate $\mathbf{p}_{a}$ onto the $k$-th relaxation mode is computed via the inner product with the left eigenvector:
\begin{equation}
\pi^{(\mathrm{phys})}_{ak}=\langle\tilde{\boldsymbol{\phi}}_{k},\mathbf{p}_{a}\rangle=\sum_{p}w^{-1}(p)\phi_{k}(p)\frac{w(p)}{W_{a}}=\frac{1}{W_{a}}\sum_{p\in\mathcal{M}_{a}}\phi_{k}(p),\label{eq_phys_projection}
\end{equation}
where $W_{a}=\sum_{p\in\mathcal{M}_{a}}w(p)$ is the partition function of the basin, and $\mathbf{p}_{a}$ is the vector of the conditional equilibrium distribution with components $p_{a}(p)=w(p)/W_{a}$ for $p\in\mathcal{M}_{a}$ and $0$ otherwise. These coefficients are precisely the physical projections and should be used in the formula for the Mahalanobis distance.

Computing the eigenvectors of the non-symmetric matrix $K$ for large sparse systems is numerically unstable and computationally expensive. Therefore, we use the symmetric matrix $S=W^{-1/2}KW^{1/2}$, whose eigenvectors $\boldsymbol{\psi}_{k}$ are related to the right eigenvectors of $K$ by the similarity transformation:
\begin{equation}
\boldsymbol{\psi}_{k}=W^{-1/2}\boldsymbol{\phi}_{k}\;\Longleftrightarrow\;\boldsymbol{\phi}_{k}=W^{1/2}\boldsymbol{\psi}_{k}.\label{eq_psi_phi_relation}
\end{equation}
Substituting this relation into formula (\ref{eq_phys_projection}), we obtain $\pi^{(\mathrm{phys})}_{ak}=\frac{1}{W_{a}}\sum_{p\in\mathcal{M}_{a}}w^{1/2}(p)\psi_{k}(p).$ Therefore, the physical projection is expressed through the standard Euclidean inner product in the symmetrized basis:
\begin{equation}
\pi^{(\mathrm{phys})}_{ak}=\langle\mathbf{u}_{a},\boldsymbol{\psi}_{k}\rangle,\label{eq_final_projection}
\end{equation}
where $\mathbf{u}_{a}$ is a vector with components $u_{a}(p)=w^{1/2}(p)/W_{a}$ for $p\in\mathcal{M}_{a}$ and $0$ otherwise. Thus, using the vectors $\mathbf{u}_{a}$ together with the eigenvectors $\boldsymbol{\psi}_{k}$ of matrix $S$ yields exactly the same numerical values of the projections as using the left/right pair of eigenvectors of the original matrix $K$.

Let us explain the physical meaning of these projections. The eigenvector $\boldsymbol{\psi}_{k}$ describes the $k$-th relaxation mode of the system, which characterizes the collective motion of the probability mass between structures occurring with characteristic time $\tau_{k}=-1/\lambda_{k}$. The projection $\langle\mathbf{u}_{a},\boldsymbol{\psi}_{k}\rangle$ shows how strongly the basin $\mathcal{M}_{a}$ participates in this collective process: a large projection in magnitude means that the structures of this basin make a significant contribution to the probability redistribution described by the mode $\boldsymbol{\psi}_{k}$. If two basins have close projections onto the slow modes (small $|\lambda_{k}|$), this means that they participate in the same slow collective processes, and transitions between them occur relatively quickly. Conversely, a large difference in projections onto the slow modes indicates that the basins are separated by high energy barriers.

Formally, let $N=|\mathcal{X}'|$ be the number of structures in the set $\mathcal{X}'$ obtained after generation and removal of duplicates. For each basin of attraction $\mathcal{M}_{a}$, its conditional equilibrium distribution vector $\mathbf{p}_{a}\in\mathbb{R}^{N}$ is defined with components
\[
p_{a}(p)=\begin{cases}
w(p)/W_{a}, & \text{if }p\in\mathcal{M}_{a},\\
0, & \text{otherwise},
\end{cases}
\]
where $W_{a}=\sum_{q\in\mathcal{M}_{a}}w(q)$ is the partition function of the basin. Then we construct the vector $\mathbf{u}_{a}$ with components $u_{a}(p)=w^{1/2}(p)/W_{a}$ for $p\in\mathcal{M}_{a}$ and project it onto the basis of eigenvectors of matrix $S$. The projection onto the $k$-th eigenvector is $\langle\mathbf{u}_{a},\boldsymbol{\psi}_{k}\rangle=\sum_{p}u_{a}(p)\,\psi_{k}(p)$.

To construct the distance, we fix a natural number $m$, $1\le m<N$, and consider the subspace $\mathcal{V}_{m}\subset\mathbb{R}^{N}$ spanned by the first $m$ eigenvectors of matrix $S$ corresponding to the smallest non-zero eigenvalues in magnitude:
\[
\mathcal{V}_{m}=\mathrm{span}\{\boldsymbol{\psi}_{1},\boldsymbol{\psi}_{2},\dots,\boldsymbol{\psi}_{m}\}.
\]
To each basin $\mathcal{M}_{a}$, we associate a vector of physical projections onto the subspace $\mathcal{V}_{m}$:
\[
\pi_{a}=\bigl(\langle\mathbf{u}_{a},\boldsymbol{\psi}_{1}\rangle,\langle\mathbf{u}_{a},\boldsymbol{\psi}_{2}\rangle,\dots,\langle\mathbf{u}_{a},\boldsymbol{\psi}_{m}\rangle\bigr)\in\mathbb{R}^{m}.
\]
Next, we define a diagonal weight matrix $\Sigma^{-1}$ of size $m\times m$ with elements
\[
(\Sigma^{-1})_{kk}=\frac{1}{|\lambda_{k}|},\quad k=1,\dots,m,
\]
and $(\Sigma^{-1})_{kl}=0$ for $k\neq l$. Then the square of the distance between basins $\mathcal{M}_{a}$ and $\mathcal{M}_{b}$ is defined as the weighted Euclidean norm of the difference of their physical projections:
\[
D^{2}(\mathcal{M}_{a},\mathcal{M}_{b})=(\pi_{a}-\pi_{b})^{T}\Sigma^{-1}(\pi_{a}-\pi_{b})=\sum^{m}_{k=1}\frac{1}{|\lambda_{k}|}\,\left(\langle\mathbf{u}_{a},\boldsymbol{\psi}_{k}\rangle-\langle\mathbf{u}_{b},\boldsymbol{\psi}_{k}\rangle\right)^{2}.
\]
For $a=b$, we set $D(\mathcal{M}_{a},\mathcal{M}_{a})=0$.

The choice of a diagonal matrix with elements $1/|\lambda_{k}|$ as $\Sigma^{-1}$ has a clear physical interpretation. The quantity $|\lambda_{k}|$ characterizes the rate of the $k$-th relaxation process: the smaller $|\lambda_{k}|$, the slower the process and the higher the energy barrier separating the basins in this process. Accordingly, the difference in projections onto a slow mode gives a larger contribution to the distance, reflecting the fact that transitions between basins with different projections onto slow modes are hindered by high barriers. Since the frequency factor $\nu_{0}$ has the dimension {[}time{]}$^{-1}$, the eigenvalues $\lambda_{k}$ also have the dimension {[}time{]}$^{-1}$, and the quantities $1/|\lambda_{k}|$ in the formula for $D^{2}$ have the dimension of time. Consequently, the quantity $D$ defined above has the dimension {[}time{]}$^{1/2}$. For checking ultrametricity, only the relative values of the distances are important, and the absolute scale is irrelevant: ultrametricity is invariant under multiplication of the metric by a positive constant. To bring the distance to the dimension of energy, matching the dimension of free energy and energy barriers, we multiply it by $RT\sqrt{\nu_{0}}$, where $\nu_{0}$ is the previously introduced frequency factor (the pre-exponential factor in the Kramers formula (\ref{Kr}), having dimension {[}time{]}$^{-1}$):
\[
\widetilde{D}(\mathcal{M}_{a},\mathcal{M}_{b})=RT\cdot\sqrt{\nu_{0}}\cdot D(\mathcal{M}_{a},\mathcal{M}_{b}).
\]
In what follows, for brevity, we will denote the normalized distance $\widetilde{D}$ simply as $D$.

On the set of basins of attraction $\mathcal{A}(S)$, we introduce an equivalence relation $\sim_{m}$ defined by the equality of physical projections onto the subspace $\mathcal{V}_{m}$:
\[
\mathcal{M}_{a}\sim_{m}\mathcal{M}_{b}\;\Longleftrightarrow\;\pi_{a}=\pi_{b}\;\Longleftrightarrow\;D(\mathcal{M}_{a},\mathcal{M}_{b})=0.
\]
Denote by $\mathcal{A}(S)/\!\sim_{m}$ the quotient set of equivalence classes under this relation.

\begin{theorem}
\par The function $D$ induces a metric on the quotient set $\mathcal{A}(S)/\!\sim_{m}$.\par
\end{theorem}
\begin{proof}
Let us verify the axioms of a metric on the quotient set. Let $[\mathcal{M}_{a}]$, $[\mathcal{M}_{b}]$, $[\mathcal{M}_{c}]$ be equivalence classes.

Non-negativity. The property $D(\mathcal{M}_{a},\mathcal{M}_{b})\ge0$ is obvious from the definition, as all terms in the sum are non-negative.

Non-degeneracy. From the definition of the distance $D$, it follows that $D(\mathcal{M}_{a},\mathcal{M}_{b})=0$ if and only if $\pi_{a}=\pi_{b}$ for all $k=1,\dots,m$. By the definition of the relation $\sim_{m}$, this means $\mathcal{M}_{a}\sim_{m}\mathcal{M}_{b}$, which is equivalent to the equality of classes $[\mathcal{M}_{a}]=[\mathcal{M}_{b}]$. Consequently, on equivalence classes, the function $D$ is strictly positive for distinct classes.

Symmetry. The property $D(\mathcal{M}_{a},\mathcal{M}_{b})=D(\mathcal{M}_{b},\mathcal{M}_{a})$ follows from the square of the difference in the definition and the symmetry of the relation $\sim_{m}$.

Triangle inequality. Let $\mathcal{M}_{a},\mathcal{M}_{b},\mathcal{M}_{c}$ be representatives of three equivalence classes. For each $k=1,\dots,m$, denote $a_{k}=\langle\mathbf{u}_{a},\boldsymbol{\psi}_{k}\rangle$, $b_{k}=\langle\mathbf{u}_{b},\boldsymbol{\psi}_{k}\rangle$, $c_{k}=\langle\mathbf{u}_{c},\boldsymbol{\psi}_{k}\rangle$, $u_{k}=\sqrt{\omega_{k}}\,|a_{k}-c_{k}|$, $v_{k}=\sqrt{\omega_{k}}\,|c_{k}-b_{k}|$, where $\omega_{k}=1/|\lambda_{k}|$. Then $D(\mathcal{M}_{a},\mathcal{M}_{c})=\|u\|$, $D(\mathcal{M}_{c},\mathcal{M}_{b})=\|v\|$, and $D(\mathcal{M}_{a},\mathcal{M}_{b})^{2}=\sum^{m}_{k=1}\omega_{k}|a_{k}-b_{k}|^{2}$. By the triangle inequality for the absolute value, we have $|a_{k}-b_{k}|\le|a_{k}-c_{k}|+|c_{k}-b_{k}|$. Hence
\[
D(\mathcal{M}_{a},\mathcal{M}_{b})^{2}=\sum^{m}_{k=1}\omega_{k}|a_{k}-b_{k}|^{2}\le\sum^{m}_{k=1}\omega_{k}(|a_{k}-c_{k}|+|c_{k}-b_{k}|)^{2}
\]
\[
=\|u\|^{2}+\|v\|^{2}+2\sum^{m}_{k=1}u_{k}v_{k}\le\|u\|^{2}+\|v\|^{2}+2\|u\|\|v\|=(\|u\|+\|v\|)^{2},
\]
where the last inequality is the Cauchy-Bunyakovsky inequality $\sum_{k}u_{k}v_{k}\le\|u\|\|v\|$. Taking the square root, we obtain $D(\mathcal{M}_{a},\mathcal{M}_{b})\le D(\mathcal{M}_{a},\mathcal{M}_{c})+D(\mathcal{M}_{c},\mathcal{M}_{b})$. Since the value of $D$ does not depend on the choice of representatives within the classes of equivalence, the triangle inequality holds on the quotient set $\mathcal{A}(S)/\!\sim_{m}$.

Thus, all four axioms of a metric are satisfied on the quotient set $\mathcal{A}(S)/\!\sim_{m}$.

Note that for $m=N-1$, the subspace $\mathcal{V}_{m}$ coincides with the orthogonal complement to the stationary distribution, and since the conditional distribution vectors of the basins are normalized and have non-overlapping supports, the relation $\sim_{m}$ is trivial (each class consists of a single element), so that $D$ in this case is a metric on the set $\mathcal{A}(S)$ itself. For $m<N-1$, different basins can fall into the same equivalence class, and $D$ is a metric precisely on the quotient set.
\end{proof}

We emphasize that the metric $D$ is not automatically an ultrametric. Indeed, the definition of $D$ uses a sum of squares of differences of projections with positive weights, which represents a Euclidean metric in the weighted space of projections. The Euclidean metric does not necessarily satisfy the strong triangle inequality $d(x,y)\le\max\{d(x,z),d(z,y)\}$, and it generally does not hold. Therefore, checking the ultrametricity of $D$ is a meaningful task, the result of which is determined by the real structure of the energy landscape.

\section{Filtering of noise modes}

\label{sec_filter}

When numerically diagonalizing large matrices with exponentially small elements, the following problem arises. The elements $S_{pq}$ for structures with large energy differences become vanishingly small (e.g., for a difference $|G(p)-G(q)|=10$ kcal/mol and $RT\approx0.6$ kcal/mol, the exponential $\exp\left(-\frac{|G(p)-G(q)|}{2RT}\right)$ is on the order of $10^{-4}$; for a difference of $20$ kcal/mol, it becomes on the order of $10^{-8}$; for a difference of $30$ kcal/mol, it is on the order of $10^{-14}$). As a result, the matrix $S$ contains elements differing by many orders of magnitude -- from values of order unity to values comparable to machine epsilon (for double precision $\varepsilon_{\text{mach}}\approx2.2\times10^{-16}$, which corresponds to approximately 16 significant decimal digits). When computing the eigenvalues of such a matrix, some of them have an absolute value comparable to machine precision and cannot be considered reliable. These modes represent numerical noise with no physical meaning. If such modes are included in the sum when calculating the Mahalanobis distance, their weights $1/|\lambda_{k}|$ become enormous, and they completely dominate the distance, rendering the result physically meaningless. To solve this problem, we use an algorithm for automatic filtering of noise modes based on the search for a spectral gap.

The algorithm uses not an absolute threshold for the absolute values of eigenvalues, but a relative criterion: the ratios $|\lambda_{(k)}|/|\lambda_{(k-1)}|$ of consecutive eigenvalues ordered by increasing absolute value are analyzed. If this ratio exceeds a given threshold $\theta$, it indicates a sharp jump in the spectrum (a spectral gap), and all modes with smaller indices are discarded as noise. Since the criterion is based on ratios rather than absolute values, it does not depend on the overall scale of the matrix and does not require threshold tuning for a specific system.

Let us describe this criterion formally in the form of an algorithm. Suppose $m_{\text{req}}$ eigenvalues of matrix $S$ with the smallest absolute values are computed. The first step of the algorithm is to sort these eigenvalues in increasing order of absolute value: $|\lambda_{(0)}|\le|\lambda_{(1)}|\le\dots\le|\lambda_{(m_{\text{req}}-1)}|$. The value $\lambda_{(0)}$ corresponds to the stationary distribution and should be equal to zero to within machine precision. Next, for each $k=1,2,\dots,m_{\text{req}}-1$, we compute the ratio
\[
\rho_{k}=\frac{|\lambda_{(k)}|}{|\lambda_{(k-1)}|}.
\]
Here, if $|\lambda_{(k-1)}|$ is less than $10^{-30}$, this value is certainly in the numerical noise region (many orders of magnitude below machine epsilon), and the formal calculation of the ratio $\rho_{k}$ becomes meaningless. In this case, we set $\rho_{k}=\infty$, which guarantees continuation of the spectral gap search and prevents false termination of the algorithm at the boundary of two noise modes.

The spectral gap is defined as the smallest $k$ for which $\rho_{k}>\theta$, where $\theta>1$ is a given threshold. All modes with indices $0,1,\dots,k-1$ are considered noise and are excluded from the summation in the definition of distance $D$. Modes with indices $k,k+1,\dots,m_{\text{req}}-1$ are considered physical and are used to compute the distance. In our numerical implementation, we use the threshold value $\theta=10^{6}$. Let us justify this choice. Among the computed eigenvalues, two groups can be distinguished that differ radically in their nature. The first group includes modes whose absolute value is comparable to machine precision ($\sim10^{-16}$) and which represent numerical noise. The second group includes physical modes corresponding to the real relaxation times of the system. The characteristic scale of physical modes can be estimated as follows. The typical free energy difference between adjacent structures within a basin is several kcal/mol, corresponding to $S$ elements of order $10^{-2}$--$10^{-1}$ and eigenvalues of order $10^{-1}$--$10^{0}$. Even for the slowest relaxation processes associated with transitions across moderate barriers (10--15 kcal/mol), the $S$ elements are of order $10^{-4}$--$10^{-5}$, and the corresponding eigenvalues are of order $10^{-5}$--$10^{-3}$. Thus, the typical absolute values of physical modes lie in the range $10^{-5}$--$10^{0}$, which is many orders of magnitude higher than the level of numerical noise ($\sim10^{-16}$). After sorting in increasing order of absolute value, all noise modes will be at the beginning of the list, and all physical modes at the end. At the boundary between these two parts of the list, there will be two adjacent eigenvalues, one of which is noise and the other physical. The ratio of their absolute values will be anomalously large (at least $10^{5}$--$10^{10}$) compared to the ratio of two noise or two physical modes (where the values differ by no more than a few orders of magnitude). The threshold $\theta=10^{6}$ is chosen as a sufficiently large number allowing reliable identification of this anomaly and separation of the noise part of the spectrum from the physical part.

In the event that a spectral gap is not detected (all ratios $\rho_{k}<\theta$ for $k=1,\dots,m_{\text{req}}-1$), the algorithm retains only the first non-zero mode. Such a case may indicate an insufficient number of requested modes or a specific structure of the spectrum and requires additional verification.

Note that the algorithm we use is based on the standard procedure for determining the effective rank of a matrix, described in the literature on numerical methods of linear algebra \cite{Golub}. A similar approach, based on searching for a sharp change in the ratio of consecutive eigenvalues, is used in principal component analysis to separate significant directions from noise \cite{Jolliffe}.

\section{Verification of ultrametricity}

\label{sec_ultr}

Let $M$ be a finite metric space with metric $d$ and $|M|=K$ elements. The most direct way to verify the ultrametricity of $M$ is to iterate over all triples of elements and check the strong triangle inequality for each triple. However, when checking ultrametricity, one should distinguish the so-called ``trivial ultrametricity'' from ``nontrivial ultrametricity''. Trivial ultrametricity arises when the distances between any three points in a metric space are equal. Such ultrametricity can occur in systems without real hierarchy, for example, when randomly choosing points in a high-dimensional space \cite{zubarev2014,zubarev2017}.

For a triple of distinct elements $(x,y,z)$, we sort the three pairwise distances in ascending order: $d_{\min}\le d_{\text{med}}\le d_{\max}$. We call a triple $(x,y,z)$ trivially ultrametric if the equality $d_{\min}=d_{\text{med}}=d_{\max}$ holds. We call a triple $(x,y,z)$ nontrivially ultrametric if the conditions $d_{\text{med}}=d_{\max}$, $d_{\min}<d_{\text{med}}$ hold. A triple that satisfies neither the condition of trivial nor nontrivial ultrametricity is called non-ultrametric.

It follows directly from the definition of an ultrametric that the space $M$ is ultrametric if and only if all triples of its elements are either trivially or nontrivially ultrametric. However, a situation is possible where the condition of nontrivial or trivial ultrametricity holds not for all triples, but for a substantial part of them. Therefore, we introduce a quantity called the ``degree of ultrametricity'' of the space $M$ \cite{bikulov_zubarev2026}. The degree of nontrivial ultrametricity of the space $M$ is the quantity
\[
u=\frac{|\mathcal{T}_{\text{nt}}|}{|\mathcal{T}|}\times100\%,
\]
where $\mathcal{T}$ is the set of all unordered triples of distinct elements from $M$, and $\mathcal{T}_{\text{nt}}\subset\mathcal{T}$ is the subset of nontrivially ultrametric triples. The concepts of degree of trivial ultrametricity and degree of non-ultrametricity are introduced analogously.

In computational experiments, distances are always known with some error due to rounding or the approximate nature of the calculations. Moreover, in continuous metrics, exact equalities have zero probability. Therefore, it makes sense to introduce the concept of approximate ultrametricity with specified accuracy thresholds. Let two positive numbers $\varepsilon$ and $\delta$ be given such that $0<\varepsilon<\delta$. We call a triple $(x,y,z)$ trivially ultrametric with accuracy $\varepsilon$ if $\frac{d_{\max}-d_{\min}}{d_{\min}}\le\varepsilon$ (in the case $d_{\min}=0$, we assume the condition is satisfied only if $d_{\max}=0$). We also call a triple $(x,y,z)$ nontrivially ultrametric with accuracy $(\varepsilon,\delta)$ if $\frac{d_{\max}-d_{\text{med}}}{d_{\text{med}}}\le\varepsilon,\:\frac{d_{\text{med}}-d_{\min}}{d_{\text{med}}}>\delta$. In other cases, we consider the triple $(x,y,z)$ non-ultrametric with accuracy $(\varepsilon,\delta)$. Here, the degree of nontrivial ultrametricity with accuracy $(\varepsilon,\delta)$ is defined as the following quantity
\begin{equation}
u_{\varepsilon,\delta}=\frac{|\mathcal{T}_{\text{nt},\varepsilon,\delta}|}{|\mathcal{T}|}\times100\%,\label{u_nt}
\end{equation}
where $\mathcal{T}_{\text{nt},\varepsilon,\delta}\subset\mathcal{T}$ is the set of triples that are nontrivially ultrametric with accuracy $(\varepsilon,\delta)$. As $\varepsilon\to0$ and $\delta\to0$, we formally obtain the exact degree of nontrivial ultrametricity $u$; however, in continuous metrics, this limiting transition has no practical meaning, since exact equalities are almost never achieved. Therefore, in computational experiments, approximate verification with specific values of $\varepsilon$ and $\delta$ satisfying the condition $0<\varepsilon<\delta$ should be used. The specific choice of these parameters is determined by the nature of the problem being solved and the required stringency of classification.

\section{Computational scheme}

\label{sec_computational}

This section provides a complete mathematical description of the computational scheme implementing the proposed method. Readers interested only in the results may skip this section without loss of understanding. The scheme consists of the following twelve sequential stages.

Stage 1. In the first stage, precomputation of allowed nucleotide pairs is performed. Namely, for a given primary structure $S$ of length $n$, the set of all position pairs satisfying the complementarity and minimum distance conditions is constructed:
\[
\mathcal{P}_{\text{allowed}}=\{(i,j):1\le i<j\le n,\;j-i\ge h_{\min}+1,\;(S(i),S(j))\in\mathcal{C}\},
\]
where $h_{\min}$ is the minimum hairpin loop length (standard value $h_{\min}=3$). The cardinality of this set does not exceed $\dbinom{n}{2}=O(n^{2})$. Each pair $(i,j)\in\mathcal{P}_{\text{allowed}}$ is bijectively assigned a characteristic function $\mathbf{1}_{ij}:\{1,\dots,n\}\to\{0,1\}$ taking the value $1$ at positions $i$ and $j$ and $0$ at all other positions. Checking that two pairs $(i,j)$ and $(k,l)$ do not conflict (do not cross and do not share common positions) reduces to checking the condition $\mathbf{1}_{ij}(m)\cdot\mathbf{1}_{kl}(m)=0$ for all $m=1,\dots,n$.

Stage 2. At this stage, a finite subset $\mathcal{X}\subset\Omega(S)$ is generated by stochastic sampling from the Gibbs distribution. Using the ViennaRNA library \cite{Vienna}, which implements the Zuker-Stiegler and McCaskill algorithm based on dynamic programming, the structure with minimum free energy $p_{\text{MFE}}$ and its energy $G_{\min}=G(p_{\text{MFE}},T)$ are computed, as well as the full partition function $Z(S)=\sum_{p\in\Omega(S)}\exp(-G(p,T)/RT)$ over the entire space $\Omega(S)$. Then, by the stochastic backtracking method (\texttt{pbacktrack} function of the ViennaRNA library), random structures $p$ are generated, each appearing with probability proportional to its Boltzmann weight: $\mathbb{P}(p)=\exp(-G(p,T)/RT)/Z(S)$. In contrast to deterministic backtracking, which chooses the variant with minimum energy at each step, stochastic backtracking chooses one of the possible variants with probability proportional to the statistical weight of the corresponding fragment. Next, for each generated structure, its free energy $G(p,T)$ is calculated. If $G(p,T)>G_{\min}+\Delta G$, where $\Delta G$ is a given energy window, this structure is discarded. The parameter $\Delta G$ can be equal to $+\infty$, meaning no energy restriction. Duplicate structures, i.e., structures with identical sets of nucleotide pairs, are removed. The process continues until a given number $N_{\max}$ of unique structures is obtained or until the maximum number of generation attempts is exhausted.

Stage 3. At this stage, duplicates are removed. For each structure $p\in\mathcal{X}$, the set of base pairs $P(p)=\{(i,j)\in\Pi:\text{positions }i\text{ and }j\text{ are paired in }p\}$ is constructed. Structures with identical sets $P(p)$ are considered identical, and from each group of identical structures, one is retained. The result is the set $\mathcal{X}'$ of pairwise distinct structures.

Stage 4. At this stage, each structure $p\in\mathcal{X}'$ is transformed into a set of pairs $P(p)$. This representation is used to generate neighboring structures and to construct the adjacency graph in the next stage.

Stage 5. At this stage, the adjacency graph is constructed. For each structure $p\in\mathcal{X}'$, all structures adjacent to $p$ are generated according to the three types of elementary operations: deletion of one pair from $P(p)$, addition of one pair from $\mathcal{P}_{\text{allowed}}\setminus P(p)$ that does not violate the conditions of a secondary structure, and substitution of one pair from $P(p)$ with one pair from $\mathcal{P}_{\text{allowed}}\setminus P(p)$, also not violating the conditions. When generating neighbors, the precomputed set $\mathcal{P}_{\text{allowed}}$ and characteristic functions $\mathbf{1}_{ij}$ (Stage 1) are used, allowing conflicts between pairs to be checked in $O(1)$ time. If a generated structure is not contained in $\mathcal{X}'$, it is added to this set. Thus, the set $\mathcal{X}'$ is augmented to be closed under the adjacency relation. After generating all neighbors, an undirected graph $G=(\mathcal{X}',\mathcal{E})$ is constructed, where $\mathcal{E}=\{\{p,q\}:p\leftrightarrow q,\;p,q\in\mathcal{X}'\}$. After this, for each structure $p$, the set of indices of its neighbors $N(p)=\{q\in\mathcal{X}':\{p,q\}\in\mathcal{E}\}$ is formed.

Stage 6. At this stage, the connectivity of the graph $G$ is checked. To do this, a search for connected components is performed using a depth-first search algorithm, which works as follows. First, an array is formed that stores for each vertex a flag indicating whether it has been visited. Then, for each unvisited vertex $v$, a recursive procedure is launched: the vertex is marked as visited, then for each of its neighbors, if it has not yet been visited, the procedure is called recursively. All vertices visited during one such call form one connected component. After finding all connected components, they are classified into significant and noise components. A connected component is considered significant if it contains at least $\max(3,\,\alpha N)$ structures, where $N=|\mathcal{X}'|$ is the total number of unique structures in the sample, and $\alpha>0$ is a given relative threshold (in our calculations, we choose $\alpha=0.001$). Components that do not satisfy this condition are classified as noise: we consider that they arise due to the incompleteness of stochastic sampling and represent isolated clusters of small size not connected to the main part of the graph. Noise components are excluded from further spectral analysis. If only one connected component is found at this stage, the graph is connected; if several connected components are found, the graph is disconnected.

Stage 7. At this stage, local minima are searched. For each structure $p\in\mathcal{X}'$, the condition is checked: for all neighbors $q\in N(p)$, $G(q,T)>G(p,T)-\varepsilon_{\text{eq}}$ holds, where $\varepsilon_{\text{eq}}$ is a small threshold accounting for possible numerical errors in free energy calculation (in our calculations, $\varepsilon_{\text{eq}}=10^{-9}$ kcal/mol). Structures satisfying this condition form the set of candidates $\mathcal{S}_{\text{cand}}$.

Stage 8. At this stage, basins of attraction are determined. First, the set of candidates $\mathcal{S}_{\text{cand}}$ is constructed, which consists of all structures that do not have a neighbor with strictly lower free energy (taking into account the threshold $\varepsilon_{\text{eq}}$). On this set, a graph $G_{\text{cand}}=(\mathcal{S}_{\text{cand}},\mathcal{E}_{\text{cand}})$ is defined, in which two structures $p,q\in\mathcal{S}_{\text{cand}}$ are connected by an edge if they are neighbors in the original graph $G$ and $|G(p,T)-G(q,T)|<\varepsilon_{\text{eq}}$. The connected components of this graph are called attractor points $A_{1},\dots,A_{R}$. If a component consists of one structure, the attractor point is a strict local minimum. If a component contains more than one structure, then the attractor point represents a plateau, which is a connected set of structures with equal (to within the threshold $\varepsilon_{\text{eq}}$) free energy, none of which has a neighbor with lower energy. The depth-first search algorithm described in Stage 6 is used to identify connected components. Then, for each structure $p\in\mathcal{X}'$, gradient descent is performed. A sequence $\{p_{k}\}^{\infty}_{k=0}$ is constructed: $p_{0}=p$; if $p_{k}$ belongs to some attractor point $A_{i}$, the descent terminates; otherwise, among all neighbors $q\in N(p_{k})$ for which $G(q,T)<G(p_{k},T)-\varepsilon_{\text{eq}}$, the one with the minimal value of $G(q,T)$ is chosen; in case of equal energies, the neighbor with the smallest index in a fixed numbering is chosen. Due to the finiteness of $\mathcal{X}'$ and the strict decrease of energy at each step, the sequence converges in a finite number of steps. The structure $p$ is assigned to the attractor point to which the sequence led. The basin of attraction $\mathcal{M}_{i}$ corresponding to point $A_{i}$ is the set of all structures whose descent ends in $A_{i}$. The family $\{\mathcal{M}_{i}\}^{R}_{i=1}$ forms a partition of $\mathcal{X}'$. This approach guarantees that plateaus are treated as unified attractor points and are not fragmented into a multitude of degenerate basins.

Stage 9. At this stage, for each basin $\mathcal{M}_{i}$, its partition function is computed:
\[
W_{i}=\sum_{p\in\mathcal{M}_{i}}\exp\left(-\frac{G(p,T)-G_{\min}}{RT}\right),
\]
where $G_{\min}=\min_{p\in\mathcal{X}'}G(p,T)$. Subtracting $G_{\min}$ in the exponent prevents numerical overflow when working with real numbers without changing the ratios of the partition functions.

Stage 10. At this stage, basins are filtered. Basins whose size $|\mathcal{M}_{i}|$ is less than a given threshold $k_{\min}$ are excluded from further analysis as statistically insignificant. If after this the number of remaining basins exceeds a given maximum value $K_{\max}$, $K_{\max}$ basins with the largest values of $W_{i}$ are selected from them. Such selection is based on the thermodynamic principle: the most probable basins (with the largest partition function) make the dominant contribution to the equilibrium properties of the system. As a result, a set of basins $\{\mathcal{M}_{1},\dots,\mathcal{M}_{K}\}$ is obtained, where $K\le K_{\max}$. The choice of $k_{\min}$ is determined by the following considerations. For small $k_{\min}$, isolated structures that do not form physically significant basins are included in the analysis. These structures often end up in separate connected components, which leads to high sample fragmentation ($f_{\text{inter}}\to1$) and makes the ultrametricity estimation unreliable. For large $k_{\min}$, many real basins of small size are discarded, and the analysis is limited to only a few largest basins, impoverishing the picture of the landscape. The optimal value depends on the sample size $N$ and the degree of graph fragmentation. Based on numerical experiments, it is recommended to choose $k_{\min}$ in the range from $5$ to $20$; the specific value should be selected so that the fraction of inter-component triples $f_{\text{inter}}$ does not exceed $0.3$--$0.5$, which ensures sufficient connectivity of the graph of significant basins. The choice of $K_{\max}$ is dictated by the computational costs for spectral decomposition and ultrametricity testing. The number of basin triples grows as $\binom{K}{3}\sim K^{3}/6$, and for $K>500$, iterating over all triples becomes resource-intensive. On the other hand, too small a $K_{\max}$ provides insufficient statistics for a reliable estimation of the degree of ultrametricity. Our recommended range is: $K_{\max}=100$--$500$. In this work, the values $k_{\min}=5$ and $K_{\max}=500$ were used.

Stage 11. At this stage, processing by connected components of the graph $G$, found in Stage 6, is performed. The disconnectedness of the graph can arise for two reasons. First, during stochastic sampling, the energy window $\Delta G$ may not be wide enough to include structures through which paths between different regions of the landscape pass. Second, the topology of the structure graph itself at a given $\Delta G$ may be objectively disconnected: there exist pairs of basins between which any path in the structure space must pass through configurations with energy exceeding $G_{\min}+\Delta G$. It is impossible to distinguish between these two situations within a single sample; however, the analysis by connected components we used guarantees that distances are computed only between those basins that are known to belong to the same connected component at a given $\Delta G$. In this case, the matrix $S$ built for the entire set $\mathcal{X}'$ is block-diagonal and has zero eigenvalues (one for each connected component). Since kinetic transitions are possible only between structures within one connected component, the distance between basins belonging to different components is physically undefined. Therefore, each connected component is processed separately: its own matrix $S$ having exactly one zero eigenvalue is built for it, and the spectral distance between basins within this component is calculated. Suppose $C$ connected components are found with sizes $N_{1},N_{2},\dots,N_{C}$ sorted in descending order. If the graph is connected ($C=1$), the entire set $\mathcal{X}'$ is considered as one component, and processing proceeds to Stage 12. If the graph is disconnected ($C>1$), for each connected component $c$, the set of basins having at least one structure in this component is determined. If the number of such basins $K_{c}$ is less than three, the component is skipped, since at least three elements are required for ultrametricity testing. Further, for each processed significant component with index $c$, the degree of nontrivial ultrametricity $u_{c}$ is calculated (over all triples of basins belonging to this component). The final degree of ultrametricity for the entire sequence is determined as the weighted average over the number of basin triples in significant components:
\[
u_{\text{weighted}}=\frac{\sum_{c\in\mathcal{C}_{\text{sig}}}u_{c}\cdot\binom{K_{c}}{3}}{\sum_{c\in\mathcal{C}_{\text{sig}}}\binom{K_{c}}{3}},
\]
where $\mathcal{C}_{\text{sig}}$ is the set of significant components containing at least three basins after filtering, and $K_{c}$ is the number of such basins in component $c$. To estimate the degree of sample fragmentation, the index $f_{\text{inter}}$ is introduced, which is the fraction of basin triples belonging to different significant connected components. If $K_{\text{sig}}=\sum_{c\in\mathcal{C}_{\text{sig}}}K_{c}$ is the total number of basins in all significant components after filtering, then
\[
f_{\text{inter}}=1-\frac{\sum_{c\in\mathcal{C}_{\text{sig}}}\binom{K_{c}}{3}}{\binom{K_{\text{sig}}}{3}}.
\]
The quantity $f_{\text{inter}}$ takes values from $0$ (all significant basins belong to one component) to values close to $1$ (significant basins are distributed over many isolated components). For large values of $f_{\text{inter}}$, the quantity $u_{\text{weighted}}$ is computed only over a small fraction of triples belonging to one component, which must be taken into account when interpreting the results: a high value of $u_{\mathrm{nt}}$ with a large $f_{\text{inter}}$ indicates hierarchy within individual connected components but does not allow conclusions about the global organization of the energy landscape.

Stage 12. At this stage, for each connected component selected in Stage 11, the transition rate matrix is constructed, spectral decomposition is performed, and Mahalanobis distances are computed. First, the structures belonging to this component are extracted. Suppose the component contains $N_{c}$ structures. A bijection $\varphi:\{1,\dots,N_{c}\}\to\mathcal{X}'_{c}$ is constructed, where $\mathcal{X}'_{c}\subset\mathcal{X}'$ is the set of structures of the component. For each basin $\mathcal{M}_{i}$, the structures that fell into this component are determined: $\mathcal{M}_{i,c}=\mathcal{M}_{i}\cap\mathcal{X}'_{c}$. If $\mathcal{M}_{i,c}\neq\varnothing$, this set forms a basin in this component. Basins for which $\mathcal{M}_{i,c}=\varnothing$ are excluded. If after this fewer than three basins remain in the component, the component is skipped. Then, the adjacency graph for this component is constructed. For each structure $p\in\mathcal{X}'_{c}$, its neighbors in graph $G$ are filtered: only those also belonging to $\mathcal{X}'_{c}$ are retained. Structure indices are transformed into local ones using the bijection $\varphi$. After this, a symmetric matrix $S_{c}$ of size $N_{c}\times N_{c}$ is constructed with elements
\[
(S_{c})_{pq}=\nu_{0}\exp\left(-\frac{|G(p)-G(q)|}{2RT}\right)
\]
for each pair of adjacent structures $p,q$ in the component ($p\neq q$) and zero for all other pairs $p\neq q$. The diagonal elements are defined as $(S_{c})_{pp}=K_{pp}=-\sum_{q\neq p}K_{qp}$, where $K_{pq}$ are computed according to (\ref{Kr}). The matrix $S_{c}$ is symmetric, negative semi-definite, and sparse: the number of non-zero elements is $N_{c}+2E_{c}$, where $E_{c}$ is the number of edges in the component. Next, the spectral decomposition of the symmetric matrix $S_{c}$ is performed. The $m_{\text{req}}$ eigenvalues with the smallest absolute values and the corresponding eigenvectors are computed. The matrix $S_{c}$ is sparse: its number of non-zero off-diagonal elements equals the number of edges in the connected component of the structure graph. Computing the spectrum of such a matrix is a standard problem in numerical linear algebra \cite{Golub}. After computing the eigenvalues, automatic filtering of noise modes is performed according to the algorithm described in Section \ref{sec_filter}. If after filtering $m_{\text{phys}}\le m_{\text{req}}$ physical modes remain, then if $m_{\text{phys}}=0$, the component is skipped. Then, for each basin $\mathcal{M}_{a}$ in this component, a vector $\mathbf{u}_{a}\in\mathbb{R}^{N_{c}}$ is constructed with components $u_{a}(p)=w^{1/2}(p)/W_{a}$ for $p\in\mathcal{M}_{a}$ and $u_{a}(p)=0$ otherwise, where $w(p)=\exp(-G(p,T)/RT)$, and $W_{a}$ is the basin's partition function (accounting for the shift by $G_{\min}$, which only leads to an overall scaling of the vector and does not affect the degree of ultrametricity). Physical projections $\pi_{ak}=\langle\mathbf{u}_{a},\boldsymbol{\psi}_{k}\rangle$ are computed for all $a=1,\dots,K_{c}$ and $k=1,\dots,m_{\text{phys}}$. After this, the matrix of Mahalanobis distances is computed:
\[
D_{ab}=\begin{cases}
0, & a=b,\\[8pt]
RT\sqrt{\nu_{0}}\cdot\sqrt{\sum^{m_{\text{phys}}}_{k=1}\dfrac{1}{|\lambda_{k}|}(\pi_{ak}-\pi_{bk})^{2}}, & a\neq b.
\end{cases}
\]
Next, for the matrix $D$, the degree of nontrivial ultrametricity $u_{\mathrm{nt}}\equiv u_{\varepsilon,\delta}$ is computed according to the definition (\ref{u_nt}) of Section \ref{sec_ultr}.

A central question arising in the interpretation of the obtained values $u_{\mathrm{nt}}$ is the question of their statistical significance, i.e., to what extent the observed partial nontrivial ultrametricity is specific to the studied biological sequences and is not a typical property of some class of random objects. The very formulation of this question, however, is ambiguous and requires a preliminary strict definition of what is meant by ``random ultrametricity'' in the context of this problem. The difficulty lies in the fact that a possible conclusion about the absence of statistically significant differences between the degree of ultrametricity of real RNAs and the degree of ultrametricity of some random model can have two fundamentally different meanings. The first is that the observed hierarchical structure may be an artifact of the method and does not reflect any real physical property of the landscape. The second is that partial ultrametricity is a universal property of the energy landscapes of heteropolymers, arising almost inevitably for practically any, including random, sequence, and in this sense, its presence in real RNAs is not statistically anomalous, but itself remains a deep physical fact. The standard frequentist approach to testing null hypotheses is poorly suited to distinguishing between these two situations: if both the real system and the null model demonstrate a non-zero degree of ultrametricity, the result of the test ``no significant differences'' carries no useful information. Therefore, a correct analysis requires the construction of a system of null hypotheses that differ in exactly which property of the system is declared random and, accordingly, which aspect of the observed ultrametricity is subjected to testing. In this work, statistical significance testing of the observed values $u_{\mathrm{nt}}$ is implemented using two null models, which are described below. For each model, a Monte Carlo ensemble of $S$ random realizations is generated (in the software implementation, the parameter \texttt{NUM\_NULL\_SAMPLES}), for each of which the degree of nontrivial ultrametricity $u^{(s)}_{\mathrm{nt}}$ is computed, $s=1,\dots,S$. From the obtained empirical distribution of the null model, the achieved significance level (p-value) is estimated. Following the logic of permutation tests, we use a two-sided criterion based on the absolute deviation from the mean of the null ensemble:
\begin{equation}
p=\frac{1}{S}\left|\left\{ s\in\{1,\dots,S\}:\left|u^{(s)}_{\mathrm{nt}}-\langle u^{\mathrm{null}}_{\mathrm{nt}}\rangle\right|\ge\left|u^{\mathrm{real}}_{\mathrm{nt}}-\langle u^{\mathrm{null}}_{\mathrm{nt}}\rangle\right|\right\} \right|,\label{eq_pvalue}
\end{equation}
where $\langle u^{\mathrm{null}}_{\mathrm{nt}}\rangle=S^{-1}\sum^{S}_{s=1}u^{(s)}_{\mathrm{nt}}$ is the sample mean of the degree of ultrametricity in the null ensemble. The choice of a two-sided criterion is dictated by the fact that biological function may require both the presence of a pronounced hierarchy (high ultrametricity) and its absence or specific frustration of the landscape (low ultrametricity). A small value of $p$ indicates that the observed degree of ultrametricity is a statistically significant property of the given sequence, not explainable within the chosen null model, regardless of the direction of deviation. It is important to emphasize that formula (\ref{eq_pvalue}) defines a symmetric two-sided criterion: the deviation of the real value from the mean of the null distribution is considered significant if it is large in absolute value, irrespective of the sign of this deviation.

Below is a description of each of the two null models.

Null model 1 is the energy shuffle model (\texttt{energy\_shuffle}). This model tests the hypothesis $H^{(1)}_{0}$ that the observed hierarchy is conditioned exclusively by the topology of the conformational transition graph, and not by the specific distribution of free energies. When generating data within this model, the set of graph vertices $\mathcal{X}'$ and the set of edges $\mathcal{E}$ (including all topological correlations -- clustering, modularity, assortativity) are preserved unchanged; however, the free energy values $\{G(p)\}_{p\in\mathcal{X}'}$ are randomly permuted. Formally, a random permutation $\sigma$ of the structure indices is generated, and the new energy of structure $p$ is assigned the value $G^{(s)}(p)=G(\sigma(p))$. After the energy reassignment, the computational chain, uniform for all significant connected components, is fully reproduced: basins of attraction are newly determined by gradient descent on the modified landscape, the symmetrized transition rate matrix $S^{(s)}$ is constructed, its spectral decomposition is performed, Mahalanobis distances are computed, and $u^{(s)}_{\mathrm{nt}}$ is determined. This procedure preserves the density of states and the topology of the conformational space but destroys the correlation between the local structure of secondary folds and their thermodynamic stability. If $u^{\mathrm{real}}_{\mathrm{nt}}$ deviates significantly from the distribution of $H^{(1)}_{0}$, this indicates that the hierarchical organization of the landscape is determined precisely by the specific arrangement of low-energy states (thermodynamic funnels), and not merely by the geometry of the graph.

Null model 2 is the nucleotide shuffle model (\texttt{nt\_shuffle}). This model is in a sense a biological control and tests the hypothesis $H^{(2)}_{0}$ that the observed ultrametricity is determined by the specific order of nucleotides fixed by evolution, and not by a consequence of general statistical properties conditioned by nucleotide composition. In this model, random permutations of the original primary sequence $S$ are generated with exact preservation of the mononucleotide composition $(n_{A},n_{U},n_{G},n_{C})$. For each such random sequence $S^{(s)}$, all stages of the computational scheme are fully executed anew: stochastic generation of a sample of secondary structures from the Gibbs distribution, construction of the adjacency graph, identification of basins of attraction for all significant connected components, filtering, construction of the transition rate matrix, spectral decomposition, calculation of the kinetic metric, and ultrametricity verification. This model destroys not only the correlation between structure and energy but also the very specific topology of the state graph characteristic of the given biological sequence. If $u^{\mathrm{real}}_{\mathrm{nt}}$ deviates significantly from the distribution of $H^{(2)}_{0}$, this serves as a direct indication that the specific order of nucleotides leads to the formation of a more pronounced (or, conversely, suppressed) hierarchical structure of the landscape than that typical for a random order at the same chemical composition. The absence of significant differences would mean that the degree of hierarchical organization is determined mainly by nucleotide composition, and not by the specific order of nucleotides.

These models process all significant connected components of the structure graph uniformly with the main computation stage (Section \ref{sec_computational}), with the same threshold $\alpha$ and subsequent weighted averaging of $u_{\mathrm{nt}}$ over components. This guarantees that the observed differences between the real system and the null model are not a consequence of different processing of disconnected graphs. For the \texttt{energy\_shuffle} model, the same sample of structures generated for the real sequence is used, which excludes variability associated with the stochastic nature of generation. The \texttt{nt\_shuffle} model, on the contrary, generates a sample anew for each realization, which increases the variance of the null distribution but makes the test the most stringent.

\section{Computational results}

\label{sec_results}

In this section, we present the results of a numerical study of the ultrametric properties of the energy landscapes of RNA secondary structures using the developed computational scheme. It should be emphasized that in this work, we did not aim to establish a connection between the functional properties of RNA and the degree of nontrivial ultrametricity of their energy landscape. The results we present are illustrative in nature and demonstrate the performance of the proposed method on real sequences with experimentally determined thermodynamic parameters.

All calculations were performed with the following parameter values. The temperature $T=37^{\circ}\mathrm{C}$ (310.15 K) corresponds to physiological conditions. The minimum hairpin loop length $h_{\min}=3$ is the standard value used in the nearest neighbor model. The energy window $\Delta G=50$ kcal/mol was chosen as sufficiently wide to include structures separated by moderate barriers, while at the same time limiting the sample size to reasonable limits. The maximum sample size $N_{\max}=100000$ is limited by our computational resources available for adjacency graph construction and spectral decomposition. The threshold $\varepsilon_{\text{eq}}=10^{-9}$ kcal/mol in determining attractor points (Stages 7, 8) was chosen negligibly small compared to typical free energy differences ($\sim0.1$--$1$ kcal/mol) to distinguish between degenerate and non-degenerate minima, while allowing for possible numerical errors in calculating $G(p,T)$. The minimum basin size $k_{\min}=5$ filters out basins containing only one or a few structures that do not form statistically representative macrostates. The maximum number of basins $K_{\max}=500$ limits the computational costs for spectral decomposition and ultrametricity verification (the number of basin triples grows as $K^{3}/6$), while retaining sufficient statistics for a reliable estimation of $u_{\mathrm{nt}}$. The relative threshold for classifying connected components $\alpha=0.001$ (Stage 6) with typical sample sizes $N\sim10^{4}$--$10^{5}$ means that $\alpha N$ ranges from $10$ to $100$ structures, and components of smaller size are excluded from the analysis. This choice is dictated by the following considerations. On the one hand, components whose size is negligibly small relative to the total sample volume do not allow a sufficient number of basins of attraction to be formed for a reliable estimation of ultrametricity within the component. On the other hand, overestimating the threshold leads to an unjustified loss of data. The value $\alpha=0.001$ was chosen as a compromise. An investigation of the sensitivity of the results to variation of $\alpha$ is beyond the scope of the present work. The ultrametricity check parameters were chosen as $\varepsilon=0.05$, $\delta=0.1$ (under the condition $\varepsilon<\delta$, see Section \ref{sec_ultr}). This choice is heuristic and reflects the degree of accuracy adopted in the work with which two distances are considered equal or different. No investigation of the sensitivity of the results to variation of $\varepsilon$ and $\delta$ was conducted. The spectral gap threshold $\theta=10^{6}$, separating noise and physical modes, is justified in Section \ref{sec_filter}. The number of requested eigenmodes $m_{\text{req}}=50$ was chosen as a compromise between the completeness of the landscape description and computational costs. The frequency factor $\nu_{0}$ in the Kramers formula (\ref{Kr}) enters as a common multiplier into all eigenvalues $\lambda_{k}$ of matrix $S$, and, consequently, into the weights $1/|\lambda_{k}|$ in the distance definition. When checking ultrametricity, only the relative values of distances are important, and the overall scale does not affect the satisfaction or violation of inequalities; therefore, the specific value of $\nu_{0}$ is immaterial. The number of null model realizations $S=100$ (Section \ref{sec_computational}) was chosen as a compromise between the accuracy of p-value estimation and the computational costs of a full landscape regeneration for each realization.

To obtain an overall picture, two large-scale computational experiments were conducted on samples of 100 sequences each. In the first experiment, 100 reference small nuclear RNAs from the NCBI database \cite{refseq} with lengths from 60 to 90 nucleotides, randomly selected, were studied. In the second experiment, 100 random sequences of the same length were generated with equiprobable ($p=0.25$) selection of each of the four nucleotides. For each sequence, the full computational scheme was executed once: stochastic generation of structures, construction of the adjacency graph, identification of basins of attraction, construction of the Kramers transition rate matrix for each connected component, spectral decomposition, and calculation of the kinetic metric with subsequent ultrametricity verification.

The results for the sample of 100 reference RNAs are as follows. The degree of nontrivial ultrametricity $u_{\mathrm{nt}}$ varies in the range from $43.86\%$ to $89.06\%$ with a mean value of $66.68\%$ and a sample standard deviation of $9.13\%$. The fraction of trivially ultrametric triples is negligibly small for all sequences (mean $0.06\%$), while the fraction of non-ultrametric triples averages $33.26\%$. The number of structures in the sample varied from $5\,786$ to $100\,000$, and the number of basins after filtering ranged from $48$ to $1\,678$. The fragmentation index $f_{\text{inter}}$ demonstrated considerable spread from $0.0000$ to $0.9888$ with a mean of $0.4566$, which indicates varying degrees of connectivity of the structure graph for different sequences.

The results for random RNAs are qualitatively similar. The value of $u_{\mathrm{nt}}$ lies in the range from $45.79\%$ to $91.68\%$, with a mean value of $68.27\%$ and a standard deviation of $8.49\%$, which is slightly higher than the mean for the reference RNAs. The fraction of non-ultrametric triples averages $31.67\%$, and trivially ultrametric triples average $0.06\%$. The fragmentation index for random RNAs is on average higher ($0.6692$), indicating a greater fragmentation of the structure graph compared to reference sequences. A comparison of the two samples shows that both natural and random RNAs demonstrate similar levels of moderate and strong nontrivial ultrametricity, with the spread of $u_{\mathrm{nt}}$ values being somewhat narrower for random sequences, and the mean value practically coinciding with the mean for reference RNAs.

The distribution of $u_{\mathrm{nt}}$ values for both samples has the following character. Weak ultrametricity ($u_{\mathrm{nt}}<50\%$) is demonstrated by only 3 out of 100 reference RNAs and 4 out of 100 random ones. Moderate ultrametricity ($50\%\le u_{\mathrm{nt}}<75\%$) is demonstrated by 72 reference and 70 random RNAs. Strong ultrametricity ($u_{\mathrm{nt}}\ge75\%$) is demonstrated by 25 reference and 26 random RNAs. Thus, the overwhelming majority of both natural and random sequences (97\% and 96\%, respectively) exhibit a moderate or strong degree of nontrivial ultrametricity of the energy landscape.

For detailed analysis, six reference RNAs were selected that showed a high degree of nontrivial ultrametricity in the primary screening (more than $70\%$) with relatively low graph fragmentation ($f_{\text{inter}}<0.4$). Their identifiers and characteristics are given below:

(1) 1207899840 XR\_002461159.1 Leishmania major strain Friedlin U5 snRNA, length 72 nt;

(2) 1063862654 XR\_001930133.1 Eremothecium sinecaudum HHLSNR56, length 76 nt;

(3) 133891602 NR\_003443.1 Caenorhabditis elegans smy-1, length 77 nt;

(4) 133891621 NR\_003462.1 Caenorhabditis elegans smy-11, length 78 nt;

(5) 1395447081 NR\_157063.1 Caenorhabditis elegans smy-7, length 82 nt;

(6) 392924753 NR\_069932.1 Caenorhabditis elegans smy-6, length 83 nt.

For each of the six selected RNAs, five independent runs of the full computational procedure with different random samples of structures were performed to assess the statistical stability of the results (Table \ref{tab_main}). The parameters for structure generation and analysis were identical to those used in the large-scale experiments.

\begin{table}[H]
\centering \caption{Results of five independent runs for six reference RNAs with high nontrivial ultrametricity. All quantities except length are given in the format mean $\pm$ standard deviation.}
\label{tab_main}{\small{}{}{}}{\small{}%
\begin{tabular}{|c|c|c|c|c|c|c|c|}
\hline
{\small \#  } & {\small Length  } & {\small Structures  } & {\small Basins  } & {\small$f_{\text{inter}}$  } & {\small$u_{\mathrm{nt}}$ (\%)  } & {\small$u_{\mathrm{tr}}$ (\%)  } & {\small$u_{\mathrm{non}}$ (\%) }\tabularnewline
\hline
{\small (1)  } & {\small 72  } & {\small 32530$\pm$46  } & {\small 294$\pm$4  } & {\small 0.3704$\pm$0.0519  } & {\small 81.22$\pm$0.55  } & {\small 0.01$\pm$0.00  } & {\small 18.77$\pm$0.55 }\tabularnewline
{\small (2)  } & {\small 76  } & {\small 33439$\pm$152  } & {\small 426$\pm$5  } & {\small 0.1519$\pm$0.0294  } & {\small 81.42$\pm$1.12  } & {\small 0.01$\pm$0.01  } & {\small 18.57$\pm$1.13 }\tabularnewline
{\small (3)  } & {\small 77  } & {\small 14521$\pm$80  } & {\small 212$\pm$5  } & {\small 0.0082$\pm$0.0184  } & {\small 76.43$\pm$1.47  } & {\small 0.01$\pm$0.00  } & {\small 23.56$\pm$1.47 }\tabularnewline
{\small (4)  } & {\small 78  } & {\small 20026$\pm$90  } & {\small 203$\pm$3  } & {\small 0.0087$\pm$0.0193  } & {\small 79.98$\pm$5.44  } & {\small 0.01$\pm$0.01  } & {\small 20.01$\pm$5.43 }\tabularnewline
{\small (5)  } & {\small 82  } & {\small 12248$\pm$55  } & {\small 119$\pm$5  } & {\small 0.0647$\pm$0.1012  } & {\small 73.77$\pm$1.94  } & {\small 0.02$\pm$0.01  } & {\small 26.21$\pm$1.95 }\tabularnewline
{\small (6)  } & {\small 83  } & {\small 21793$\pm$98  } & {\small 268$\pm$7  } & {\small 0.0335$\pm$0.0370  } & {\small 75.60$\pm$2.24  } & {\small 0.01$\pm$0.01  } & {\small 24.39$\pm$2.25 }\tabularnewline
\hline
\end{tabular}}
\end{table}

The results of Table \ref{tab_main} demonstrate that all six selected RNAs stably reproduce a high degree of nontrivial ultrametricity in the range $73.77\%$--$81.42\%$, which significantly exceeds the mean value over the general sample of reference RNAs ($66.68\%$). The standard deviations of $u_{\mathrm{nt}}$ over five runs are small (from $0.55\%$ to $5.44\%$), indicating the statistical stability of the results. The fraction of trivially ultrametric triples is negligibly small for all sequences (less than $0.1\%$). The fragmentation index $f_{\text{inter}}$ for all six RNAs is relatively low (from $0.0082$ to $0.3704$), indicating high connectivity of the structure graph and allowing the obtained values of $u_{\mathrm{nt}}$ to be interpreted as a characteristic of the global, rather than local, hierarchical organization of the energy landscape.

To check the statistical significance of the observed ultrametricity, two independent tests based on different null models were conducted. The results of the nucleotide shuffle test (\texttt{nt\_shuffle}) are given in Table~\ref{tab_nt}.

\begin{table}[H]
\centering \caption{Results of the nt\_shuffle test (nucleotide permutation preserving nucleotide composition) for six reference RNAs. The null distribution was obtained from 100 realizations.}
\label{tab_nt}{\small{}{}{}}{\small{}%
\begin{tabular}{|c|c|c|c|c|c|c|}
\hline
{\small \#  } & {\small Length  } & {\small$u^{\mathrm{real}}_{\mathrm{nt}}$ (\%)  } & {\small$u^{\mathrm{null}}_{\mathrm{nt}}$ (\%)  } & {\small$u^{\mathrm{null}}_{\mathrm{tr}}$ (\%)  } & {\small$u^{\mathrm{null}}_{\mathrm{non}}$ (\%)  } & {\small$p$ }\tabularnewline
\hline
{\small (1)  } & {\small 72  } & {\small 81.22  } & {\small 67.37$\pm$8.26  } & {\small 0.07$\pm$0.11  } & {\small 32.56$\pm$8.23  } & {\small 0.0500 }\tabularnewline
{\small (2)  } & {\small 76  } & {\small 81.42  } & {\small 67.68$\pm$8.90  } & {\small 0.05$\pm$0.09  } & {\small 32.27$\pm$8.87  } & {\small 0.1200 }\tabularnewline
{\small (3)  } & {\small 77  } & {\small 76.43  } & {\small 69.56$\pm$8.54  } & {\small 0.04$\pm$0.06  } & {\small 30.40$\pm$8.53  } & {\small 0.3300 }\tabularnewline
{\small (4)  } & {\small 78  } & {\small 79.98  } & {\small 70.13$\pm$8.54  } & {\small 0.06$\pm$0.19  } & {\small 29.81$\pm$8.51  } & {\small 0.4100 }\tabularnewline
{\small (5)  } & {\small 82  } & {\small 73.77  } & {\small 67.69$\pm$9.05  } & {\small 0.05$\pm$0.12  } & {\small 32.27$\pm$9.04  } & {\small 0.3900 }\tabularnewline
{\small (6)  } & {\small 83  } & {\small 75.60  } & {\small 68.29$\pm$8.50  } & {\small 0.04$\pm$0.09  } & {\small 31.67$\pm$8.47  } & {\small 0.3500 }\tabularnewline
\hline
\end{tabular}}
\end{table}

With a full landscape regeneration for random permutations of the original sequence (preserving nucleotide composition), the average degree of nontrivial ultrametricity decreases to values in the range $67.37\%$--$70.13\%$, which is close to the average values observed for the general sample of 100 reference and 100 random RNAs. The standard deviations of the null distribution in this test are significant (from $8.26\%$ to $9.05\%$), reflecting the additional variability introduced by the full regeneration of the structural ensemble for each permutation. Two-sided p-values for all six sequences fell in the range from $0.05$ to $0.41$. At the standard significance level $\alpha=0.05$, only RNA (1) demonstrates borderline significance ($p=0.0500$), whereas for the remaining five RNAs, the null hypothesis cannot be rejected. This means that, for a fixed nucleotide composition, the degree of nontrivial ultrametricity of the energy landscape is not determined to a statistically significant degree by the specific order of nucleotides for most of the studied RNAs with high ultrametricity.

The results of the energy shuffle test (\texttt{energy\_shuffle}) while preserving the adjacency graph are given in Table~\ref{tab_energy}.

\begin{table}[H]
\centering \caption{Results of the energy\_shuffle test (energy permutation preserving the adjacency graph) for six reference RNAs. The null distribution was obtained from 100 realizations. Values at $p<0.01$ are indicated as $p<0.01$.}
\label{tab_energy}{\small{}{}{}}{\small{}%
\begin{tabular}{|c|c|c|c|c|c|c|}
\hline
{\small \#  } & {\small Length  } & {\small$u^{\mathrm{real}}_{\mathrm{nt}}$ (\%)  } & {\small$u^{\mathrm{null}}_{\mathrm{nt}}$ (\%)  } & {\small$u^{\mathrm{null}}_{\mathrm{tr}}$ (\%)  } & {\small$u^{\mathrm{null}}_{\mathrm{non}}$ (\%)  } & {\small$p$ }\tabularnewline
\hline
{\small (1)  } & {\small 72  } & {\small 81.22  } & {\small 57.55$\pm$2.67  } & {\small 0.09$\pm$0.04  } & {\small 42.36$\pm$2.68  } & {\small$<$0.01 }\tabularnewline
{\small (2)  } & {\small 76  } & {\small 81.42  } & {\small 44.91$\pm$6.55  } & {\small 0.04$\pm$0.04  } & {\small 55.05$\pm$6.55  } & {\small$<$0.01 }\tabularnewline
{\small (3)  } & {\small 77  } & {\small 76.43  } & {\small 43.42$\pm$5.12  } & {\small 0.04$\pm$0.03  } & {\small 56.54$\pm$5.11  } & {\small$<$0.01 }\tabularnewline
{\small (4)  } & {\small 78  } & {\small 79.98  } & {\small 55.20$\pm$4.75  } & {\small 0.04$\pm$0.01  } & {\small 44.76$\pm$4.76  } & {\small$<$0.01 }\tabularnewline
{\small (5)  } & {\small 82  } & {\small 73.77  } & {\small 63.49$\pm$2.12  } & {\small 0.04$\pm$0.03  } & {\small 36.47$\pm$2.12  } & {\small$<$0.01 }\tabularnewline
{\small (6)  } & {\small 83  } & {\small 75.60  } & {\small 62.66$\pm$4.48  } & {\small 0.03$\pm$0.04  } & {\small 37.31$\pm$4.48  } & {\small$<$0.01 }\tabularnewline
\hline
\end{tabular}}
\end{table}

The data in Table~\ref{tab_energy} show that under random energy permutation, the degree of nontrivial ultrametricity significantly decreases to the range $43.42\%$--$63.49\%$, which is substantially lower than the values for real RNAs ($73.77\%$--$81.42\%$). For all six sequences, the two-sided p-value was found to be less than $0.01$. This means that the observed high ultrametricity of real RNAs cannot be explained solely by the topology of the structure graph: the specific distribution of energies over the graph vertices plays a decisive role. The standard deviations of the null distribution $u^{\mathrm{null}}_{\mathrm{nt}}$ are relatively small (from $2.12\%$ to $6.55\%$), indicating the stability of the estimates.

A comparison of the results of the two null tests reveals the following hierarchy of contributions to ultrametricity. Random energy permutation while preserving the graph (\texttt{energy\_shuffle}) reduces $u_{\mathrm{nt}}$ by $10$--$35$ percentage points relative to the real values, and the reduction is highly statistically significant ($p<0.01$ for all RNAs). The null hypothesis \texttt{nt\_shuffle} regarding the destruction of nontrivial ultrametricity of RNAs with a high degree of nontrivial ultrametricity upon changing the order of nucleotides cannot be rejected. This may indicate that the main contribution to the high ultrametricity of the selected RNAs is made by the specific distribution of free energies over the vertices of the structure graph.

\section{Discussion and prospects}

\label{sec_conclusion}

The present work is a methodological study, the main result of which is the construction and validation of a tool for analyzing the ultrametricity of RNA energy landscapes. In this work, a kinetic metric between basins of attraction on the set of secondary structures has been constructed, which takes into account all possible transition paths between basins and the thermodynamic weights of these paths and is based on the spectral decomposition of the symmetrized Kramers transition rate matrix. A computational scheme for ultrametricity analysis has been developed, including an algorithm for automatic separation of physical modes from numerical noise by searching for a spectral gap and a method for processing disconnected structure graphs arising from stochastic sampling from the Gibbs distribution. The performance of the method has been numerically demonstrated in computational experiments: a screening of 100 reference and 100 random RNAs with lengths from 60 to 90 nucleotides, as well as a detailed statistical analysis of six reference RNAs with a high degree of nontrivial ultrametricity. All calculations were performed at the physiological temperature $37^{\circ}\mathrm{C}$ with an energy window of 50 kcal/mol.

The conducted numerical experiments on samples of 100 reference and 100 random RNAs allow the following main conclusions to be formulated. First, both natural and random sequences demonstrate predominantly a moderate and strong degree of nontrivial ultrametricity of the energy landscape in the range $50$--$90\%$. This result holds for the overwhelming majority of the studied RNAs (97\% reference and 96\% random) regardless of their origin and indicates that a high degree of hierarchical organization is, apparently, a universal property of the energy landscapes of RNA secondary structures, at least for sequences of the considered length. Second, individual representatives of reference RNAs demonstrate a strong degree of nontrivial ultrametricity exceeding $80\%$ and reaching $89\%$. These sequences are of particular interest as potential candidates for which the hierarchical organization of the landscape may have functional significance. Third, the results of the \texttt{energy\_shuffle} null test unequivocally indicate that, for a fixed adjacency graph, the degree of nontrivial ultrametricity is decisively determined by the distribution of free energies of the structures, and not merely by the graph topology. Fourth, the results of the \texttt{nt\_shuffle} null test show that, for a fixed nucleotide composition, the degree of nontrivial ultrametricity is not determined to a statistically significant degree by the order of nucleotides for most RNAs with high ultrametricity.

It must be emphasized that within the framework of the present work, the goal was not to link the functional properties of specific RNAs to the observed degree of nontrivial ultrametricity of their energy landscapes. The task of the work was to develop and test a method for quantitative assessment of the hierarchical organization of the energy landscape that does not generate artifactual ultrametricity, as well as to demonstrate the fact that nontrivial ultrametricity in a moderate and strong degree indeed occurs for most RNA representatives. Establishing possible connections between the degree of ultrametricity and biological function, as well as the role of evolutionary selection in shaping the observed hierarchical structure of landscapes, requires further systematic studies on substantially larger samples of sequences with known functional characteristics.

The obtained results allow the following conjecture about the physical nature of the observed high ultrametricity of the kinetic metric. It appears that in most of the studied RNAs, when accounting for possible transition paths between basins, a single path, determined by the single-linkage distance, predominates, and it is precisely this path that makes the main contribution to the kinetic metric. The contribution of all other paths is less significant and is largely suppressed. Such a situation can be realized if the energies of Markov microstates (structures) vary over fairly wide ranges and, on most paths, the edges of the graph between vertices (structures) $p$ and $q$ have a weight significantly exceeding unity, i.e., $|\Delta G|/RT\gg1$, where $\Delta G=G(p)-G(q)$. In this case, out of all contributions to transition probabilities $\sim\exp(-|\Delta G|/RT)$ along the paths, only the contribution of one path -- the min-max path -- is significant. To quantitatively illustrate this assertion, let us estimate the typical scatter of free energies in the system under consideration. At the physiological temperature $T=310.15$ K, the value $RT\approx0.617$ kcal/mol. The characteristic scatter of free energies of structures within a single basin, estimated from the samples of the conducted calculations, ranges from $5$ to $15$ kcal/mol, which gives $|\Delta G|/RT\sim8$--$24$ for most pairs of adjacent structures. The corresponding exponential factors $\exp(-|\Delta G|/RT)$ are of order $10^{-4}$--$10^{-11}$. Given such smallness of transition probabilities, the contribution of alternative (suboptimal) paths passing through a larger number of edges with large energy differences proves to be exponentially suppressed compared to the min-max path, and the kinetic metric turns out to be close to the single-linkage ultrametric, which manifests itself in high values of $u_{\mathrm{nt}}$.

On the other hand, if the energies of Markov microstates vary within limited ranges and $|\Delta G|/RT\sim1$, then other paths should also contribute to the transitions. In this case, the kinetic metric is not close to the single-linkage metric, which leads to a decrease in nontrivial ultrametricity. This is precisely what is observed for those relatively few RNAs whose degree of nontrivial ultrametricity is not high, but moderate or weak. Here, the significant reduction in the degree of nontrivial ultrametricity in RNAs with strong ultrametricity under a fixed adjacency graph and random permutation of structure energies (\texttt{energy\_shuffle} test) can be explained by the fact that such a permutation creates multiple realizations of weights on the adjacency graph in which suboptimal paths begin to be realized, making a significant contribution to the kinetic metric. In other words, random reassignment of energies destroys the specific organization of the landscape in which one min-max path dominates over the others and leads to the emergence of competing paths with comparable probabilities. Other explanations of this phenomenon are also possible, requiring further theoretical analysis.

We also note a number of limitations of the proposed method that dictate directions for its further development.

First, the Kramers formula is used in this work with the lower bound on the barrier height $\max\{G(p),G(q)\}$, which is a standard approximation for models of RNA secondary structure kinetics \cite{hofacker1994,Zuckerman,flamm2000}, in which the exact geometry of the energy surface is inaccessible. Nevertheless, the true barrier may be higher if the transition requires the simultaneous rupture of several pairs or passage through conformations with unfavorable energy. Refining the barrier estimate by explicit search for saddle points on the energy surface (e.g., using transition state optimization methods) would allow for the construction of a more accurate transition rate matrix and, as a consequence, would enhance the physical adequacy of the kinetic metric.

The issue of constructing the characteristic vectors of basins also requires commentary. In this work, to account for the non-uniformity of the probability distribution within a basin, symmetrized characteristic vectors $\mathbf{u}_{a}$ with weights proportional to the square root of the Boltzmann factor $w^{1/2}(p)$ are used directly. Such a representation is physically justified and reflects the real contribution of structures. The symmetrization of the transition rate matrix $K$, described in Section \ref{sec_spectral}, transforms the stationary distribution of the Markov process, and a correct accounting of non-equilibrium kinetics requires the introduction of exactly such symmetrized characteristic vectors. Investigation of the sensitivity of the results to variations in the method of constructing characteristic vectors (e.g., under a large spread of free energies within a single basin) is a subject for further study.

We also note that the choice of the number of requested eigenmodes $m_{\text{req}}$ is a compromise between the completeness of the landscape description and computational costs. In the performed calculations, $m_{\text{req}}=50$ was used, and all requested modes except the stationary one were classified by the filtering algorithm as physical. This indicates that a more detailed analysis may require an increase in $m_{\text{req}}$, which, in turn, necessitates the use of more efficient algorithms for spectral decomposition of sparse matrices.

The question of the stability of the kinetic metric to the volume of stochastic sampling $N$ deserves special attention. Since the complete space of structures $\Omega(S)$ is exponentially large, the matrix $S$ is always built on a finite subset $\mathcal{X}'\subset\Omega(S)$, and its spectrum depends on the composition of this subset. A rigorous proof of the convergence of the sample matrix spectrum to the spectrum of the matrix on the whole $\Omega(S)$ as $N\to|\Omega(S)|$ is not feasible; however, for practical purposes, the stability of the final values of $u_{\mathrm{nt}}$ under reasonable variation of $N$ is more important. Within the framework of this work, such stability is confirmed: for all six studied RNAs, the standard deviation of $u_{\mathrm{nt}}$ over five independent runs with $N_{\max}=100000$ does not exceed $5.5$ absolute percent (Table~\ref{tab_main}). A direct study of the dependence of $u_{\mathrm{nt}}$ on $N$ (e.g., by successively increasing the sample) is a task for further parametric analysis and lies outside the scope of this methodological work.

It is also worth noting that the most resource-intensive stage of the computational scheme outlined in Section \ref{sec_computational} is Stage 5, consisting of the construction of the adjacency graph. As can be easily estimated, the computational complexity of this Stage is $O(N\cdot P\cdot d)$, where $N$ is the number of structures, $P$ is the number of allowed pairs, and $d$ is the average degree of a vertex. For sequences of length $L>200$, this stage requires significant computational resources. Therefore, a necessary condition for applying the method to longer sequences is the development of optimized parallel computing algorithms.

Summarizing the above, it can be concluded that the proposed kinetic metric is a workable tool for analyzing the hierarchical organization of biopolymer energy landscapes, applicable to RNAs up to several hundred nucleotides in length. This method does not generate artifactual ultrametricity, is sensitive to structural features of sequences, and demonstrates stability to statistical fluctuations. The results obtained show that a high degree of nontrivial ultrametricity is a typical, rather than exceptional, property of the energy landscapes of RNA secondary structures in the studied length range. The fundamentally important next steps are the extension of the method to a wider class of sequences of varying length and functional purpose, as well as a systematic verification of the statistical significance of the observed partial ultrametricity relative to a hierarchy of different null models. These studies, relying on the methodological base developed in this work, will allow answering the question of whether the observed hierarchical organization of the landscape is a product of evolutionary selection or a fundamental property of the space of RNA secondary structures.

\section*{Conflict of Interest}

The author declares no conflict of interest.

\section*{Funding}

The study was carried out without external funding sources.

\section*{Data and Code Availability}

The source code of the program is openly available in the Zenodo
repository \cite{zubarev2026_github}.  Additional data
supporting the findings of this study are available from the
author upon reasonable request.

\end{document}